\documentclass[12pt,technote, onecolumn, draft]{IEEEtran}
\usepackage{latexsym} 
\usepackage{mathrsfs}
\usepackage{multirow}
\usepackage{amsfonts,amssymb,amsmath,amsthm,bm}
\usepackage{color}
\usepackage{cite}
\usepackage{comment}
\usepackage{graphicx}
\usepackage{caption}
\usepackage{array}
\newtheorem{theorem}{Theorem}

\newtheorem{remark}{Remark}

\newtheorem{lemma}{Lemma}

\newcommand{\RomanNum}[1]{\uppercase\expandafter{\romannumeral #1\relax}}

\begin{document}
	\title{New families of asymptotically optimal codebooks from vectorial dual-bent functions$^{\dag}$}
	
	\author{Yadi Wei, Jiaxin Wang, Fang-Wei Fu, Wenjuan Yin
		\IEEEcompsocitemizethanks{\IEEEcompsocthanksitem Yadi Wei and Fang-Wei Fu  are with the Chern Institute of Mathematics and LPMC, Nankai University, Tianjin 300071, China, Emails:  wydecho@mail.nankai.edu.cn,  fwfu@nankai.edu.cn.
			Jiaxin Wang is with the School of Mathematics, Hefei University of Technology, Hefei 230601, China, Email: wjiaxin@hfut.edu.cn. Wenjuan Yin is with the School of Mathematical Sciences, Tiangong University, Tianjin 300387, China, Email: yinwj@tiangong.edu.cn}
		
		\thanks{$^\dag$This research is supported by the National Key Research and Development Program of China (Grant No. 2022YFA1005000), the National Natural Science Foundation of China (Grant Nos. 62371259, 12411540221), the Fundamental Research Funds for the Central Universities of China (Nankai University), the Nankai Zhide Foundation.  }
	}

	\maketitle

	\begin{abstract}
		Codebooks with small maximum cross-correlation amplitudes play an important role in many applications, such as code division multiple access (CDMA) communication systems, multiple-input multiple-output (MIMO) communications, compressed sensing, and coding theory. In this paper, by using vectorial dual-bent functions, we construct several families of codebooks that asymptotically achieve the Welch bound. The maximum cross-correlation amplitudes and the distributions of the cross-correlation amplitudes of the constructed codebooks are explicitly determined. Furthermore, these codebooks have new parameters, and some of them have very small alphabet sizes.
	\end{abstract}
	
	\section{Introduction}\label{Section1}
	The set $\mathcal{C}=\{\mathbf{c}_0,\mathbf{c}_1,\dots,\mathbf{c}_{N-1}\}$, where each $\mathbf{c}_i, 0\le i\le N-1$, is a unit norm $1\times K$ complex-valued vector over an alphabet $A$, is called an $(N,K)$ codebook (also called signal set). The size of the alphabet $A$ is said to be the alphabet size of the codebook $\mathcal{C}$. In practical applications, codebooks with small alphabet size are desirable, since they are easier to implement. The  maximum cross-correlation amplitude of the $(N,K)$ codebook $\mathcal{C}$, as an important performance measure of $\mathcal{C}$, is defined as
	{\small 	\[I_\mathrm{max}(\mathcal{C})=\max_{0\le i\ne j\le N-1}|\mathbf{c}_i\mathbf{c}_j^H|,\]}
	where $\mathbf{c}_j^H$ is the conjugate transpose of the vector $\mathbf{c}_j$. Codebooks with small $I_\mathrm{max}(\mathcal{C})$
	are of great interest in many applications, including code division multiple access (CDMA) communication systems \cite{Massey}, multiple-input multiple-output (MIMO) communications \cite{Heath}, compressed sensing \cite{Li}, and coding theory \cite{Fan}. For a fixed length $K$, it is desirable to construct an $(N,K)$ codebook $\mathcal{C}$ with $N$ as large as possible and $I_\mathrm{max}(\mathcal{C})$ as small as possible. However, there is a trade-off among $N$, $K$ and $I_\mathrm{max}(\mathcal{C})$. The well-known bounds are given as follows.
	\begin{itemize}
		\item [$\bullet$] The Welch bound \cite{Welch}: For any $(N,K)$ codebook $\mathcal{C}$ with $N\ge K$, {\small \[I_{\mathrm{max}}(\mathcal{C})\ge I_{\mathrm{W}}=\sqrt{\frac{N-K}{(N-1)K}}.\]}
		\item[$\bullet$] The Levenshtein bound \cite{Kabatyanskii,Levenstein}:
		
		For any real-valued $(N,K)$ codebook $\mathcal{C}$ with $N>\frac{K(K+1)}{2}$,
		{\small \[I_{\mathrm{max}}(\mathcal{C})\ge I_{\mathrm{L}}=\sqrt{\frac{3N-K^2-2K}{(N-K)(K+2)}}.\]}
		For any complex-valued $(N,K)$ codebook $\mathcal{C}$ with $N>K^2$,
		{\small 	\[I_{\mathrm{max}}(\mathcal{C})\ge I_{\mathrm{L}}=\sqrt{\frac{2N-K^2-K}{(N-K)(K+1)}}.\]}
	\end{itemize}
	A codebook achieving equality in the Welch bound is called a maximum-Welch-bound-equality (MWBE) codebook \cite{Xia} or an equiangular tight frame \cite{Christensen}. When $N$ is much bigger than $K$, the Levenshtein bound is tighter than the Welch bound. Several families of codebooks achieving the Welch bound or the Levenshtein bound, referred to as optimal codebooks, have been presented in  \cite{Calderbank,Conway,Ding1,Ding2,Ding3,Fickus,Sarwate,Strohmer,Xia,Xiang,Zhou}. In general, it is very hard to design optimal codebooks, thus considerable attention has been devoted to asymptotically optimal codebooks. More precisely, an $(N,K)$ codebook $\mathcal{C}$ is called asymptotically optimal with respect to the Welch bound (respectively, the Levenshtein bound), if $\lim_{K\rightarrow +\infty}\frac{I_{\mathrm{max}}(\mathcal{C})}{I_{\mathrm{W}}}=1$ (respectively, $\lim_{K\rightarrow +\infty}\frac{I_{\mathrm{max}}(\mathcal{C})}{I_{\mathrm{L}}}=1$). Compared with optimal codebooks, asymptotically optimal codebooks have more flexible parameters.
	
	There are numerous constructions of asymptotically optimal codebooks proposed in the literature. Several families of codebooks  asymptotically achieving the Welch bound were constructed in \cite{Ding4,Hu,Zhang,Zhou1} based on difference sets, almost difference sets and relative difference sets. Some codebooks that asymptotically achieve the Welch bound were presented in\cite{Hong,Yu} from binary sequences. In addition, several classes of  asymptotically optimal codebooks with respect to the Welch bound were obtained via exponential sums in \cite{Heng2,Heng3,Lu,Luo1,Tan,Tian,Wu,Yan,Yin}. Furthermore, some codebooks asymptotically achieving the Welch bound or the Levenshtein bound were derived from generalized bent functions over $\mathbb{Z}_4$ \cite{Zhou2}, and vectorial dual-bent functions \cite{Heng1}.  
	
	In this paper, by employing vectorial dual-bent functions, we obtain several new families of codebooks that asymptotically achieve the Welch bound based on additive and multiplicative characters. The maximum cross-correlation amplitudes of these codebooks are determined, and the corresponding distributions of cross-correlation amplitudes are also given. The parameters of the constructed codebooks, listed in Table $\mathrm{\RomanNum{1}}$, are new to the best of our knowledge. Moreover, it is worth noting that some of these codebooks have very small alphabet sizes, which may be desirable in practical applications.
	
	The rest of the paper is organized as follows. In Section $\mathrm{\RomanNum{2}}$, we introduce the necessary preliminaries. In Section $\mathrm{\RomanNum{3}}$, we give some results on vectorial dual-bent functions. In Section $\mathrm{\RomanNum{4}}$, using vectorial dual-bent functions, we present two constructions of codebooks based on additive characters. In Section $\mathrm{\RomanNum{5}}$, using vectorial dual-bent functions, we present three constructions of codebooks based on additive and multiplicative characters. In Section $\mathrm{\RomanNum{6}}$, we make a conclusion and propose some problems for future research.
	\begin{table}[h]
		\centering
		\caption{The parameters of codebooks constructed in this paper}
		\renewcommand\arraystretch{1.1}	
		\begin{tabular}{|>{\centering\arraybackslash}m{6cm}|>{\centering\arraybackslash}m{4.9cm}|c|}
			\hline
			Parameters $(N,K)$& $I_{\mathrm{max}}$& Construction\\\hline
			$(p^n,K=(p^{n-m}-\varepsilon p^{\frac{n}{2}-m})|I|+\varepsilon \delta_{I}(0)p^{\frac{n}{2}})$, $\frac{p^m-p^s}{2}\le |I|\le \frac{p^m+p^s}{2}, 2|n,m\le \frac{n}{2}$, $s<m$, $\varepsilon=\pm1$ when $p=3$ and $\varepsilon=1$ when $p>3$ &$\frac{ p^{\frac{n}{2}}(p^m-|I|)}{p^mK}\ \text{if}\ \frac{p^m-p^s}{2}\le |I|\le \frac{p^m-1}{2}$, $\frac{ p^{\frac{n}{2}}|I|}{p^mK}\ \text{if}\ \frac{p^m+1}{2}\le |I|\le \frac{p^m+p^s}{2}$& Theorem \ref{Th1}\\\hline
			$(p^{3m},K=\frac{p^m+1}{2}(p^{2m}-p^m+1))$&$\frac{(p^m+1)(p^{\frac{m}{2}+1})}{2K}\ \text{if}\ p^m\equiv 1\ (\text{mod}\ 4)$, $\frac{(p^m+1)\sqrt{p^m+1}}{2K}\ \text{if}\ p^m\equiv 3\ (\text{mod}\ 4)$
			&Theorem \ref{Th2}\\\hline
			$(p^{3m},K=\frac{p^m-1}{2}(p^{2m}+p^m+1))$&$\frac{(p^m+1)(p^{\frac{m}{2}+1})}{2K}\ \text{if}\ p^m\equiv 1\ (\text{mod}\ 4)$, $\frac{(p^m+1)\sqrt{p^m+1}}{2K}\ \text{if}\ p^m\equiv 3\ (\text{mod}\ 4)$
			&Theorem \ref{Th2}\\\hline
			$(p^{4m},K=\frac{p^{4m}+1}{2})$&$\frac{(p^m+1)^2}{2K}$&Theorem \ref{Th3}\\\hline
			$(p^{4m},K=\frac{p^{4m}-1}{2})$&$\frac{(p^m+1)^2}{2K}$&Theorem \ref{Th3}\\\hline
			$(p^{n+m}+p^n,p^n), 2|n, m\le \frac{n}{2}$&$p^{-\frac{n}{2}}$&Theorem \ref{Th4}\\\hline
			$(p^{3m}+K,K), K=(p^m-1)(p^m\pm 1)$&$\frac{p^m+1}{K}$&Theorem \ref{Th5}\\\hline
			$(p^{3m-1}+p^{2m}-p^m, p^{2m}-p^m)$&$\frac{1}{p^m-1}$&Theorem \ref{Th6}\\\hline 
			$(p^{n+m},p^n), 2|n,m\le \frac{n}{2}, (p,m)\ne (3,1)$&$\frac{p^{\frac{n}{2}}}{(p^m-1)(p^{n-m}\pm p^{\frac{n}{2}-m})}$&Theorem \ref{Th7}, Remark \ref{Re6}\\\hline
			$(p^n(p^m-1)+K, K)$, $K=(p^m-1)(p^{n-m}\pm p^{\frac{n}{2}-m}), 2|n, m\le \frac{n}{2}$&$\frac{p^{\frac{n}{2}}}{K}$&Theorem \ref{Th8}, Remark \ref{Re6}\\\hline
			$(p^{n+2m}, p^{n+m}), 2|n, m\le \frac{n}{2}$, $(p,m)\ne (3,1)$& $\frac{p^{\frac{n+m}{2}}}{(p^m-1)^2(p^{n-m}\pm p^{\frac{n}{2}-m})}$&Theorem \ref{Th9}\\\hline
			$(p^n(p^{2m}-p^m+1),p^{n+m})$, $2|n, m\le \frac{n}{2}, (p,m)\ne (3,1)$&$\frac{p^{\frac{n+m}{2}}}{(p^m-2)(p^m-1)(p^{n-m}\pm p^{\frac{n}{2}-m})}$&Theorem \ref{Th10}\\\hline
		\end{tabular}
	\end{table}

	\section{Preliminaries}
	In this section, we introduce some necessary preliminaries that will be used in the sequel.
	\subsection{Notations}
	Let $p$ be an odd prime and $n$, $m$ be positive integers. Throughout this
	paper, we set the following notations.
	\begin{itemize}
		\item [$\bullet$]$\mathbb{F}_{p^n}$ is the finite field with $p^n$ elements.
		\item [$\bullet$] $\mathbb{F}_{p^n}^*$ is the multiplicative group of the finite field $\mathbb{F}_{p^n}$.
		\item [$\bullet$] $V_n^{(p)}$ is an $n$-dimensional vector space over $\mathbb{F}_p$.
		\item[$\bullet$] $\mathbb{F}_p^n$ is a vector space of $n$-tuples over $\mathbb{F}_p$.
		\item[$\bullet$] $\mathrm{Tr}_m^n$ denotes the trace function from $\mathbb{F}_{p^n}$ to $\mathbb{F}_{p^m}$, where $m|n$.
		\item[$\bullet$]$\langle \cdot \rangle_{n}$ is a
		(non-degenerate) inner product of $V_{n}^{(p)}$. In this paper, when $V_{n}^{(p)}=\mathbb{F}_{p}^{n}$, let $\langle \mathbf{a}, \mathbf{b}\rangle_{n}=\mathbf{a} \cdot \mathbf{b}=\sum_{i=1}^{n}a_{i}b_{i}$, where $\mathbf{a}=(a_{1}, \dots, a_{n}), \mathbf{b}=(b_{1}, \dots, b_{n})\in \mathbb{F}_{p}^{n}$; when $V_{n}^{(p)}=\mathbb{F}_{p^n}$, let $\langle a, b\rangle_{n}=\mathrm{Tr}_{1}^{n}(ab)$, where $a, b \in \mathbb{F}_{p^n}$; when $V_{n}^{(p)}=V_{n_{1}}^{(p)}\times \dots \times V_{n_{t}}^{(p)} (n=\sum_{i=1}^{t}n_{i})$, let $\langle \mathbf{a}, \mathbf{b}\rangle_{n}=\sum_{i=1}^{t}\langle a_{i}, b_{i}\rangle_{n_{i}}$, where $\mathbf{a}=(a_{1}, \dots, a_{t}), \mathbf{b}=(b_{1}, \dots, b_{t})\in V_{n}^{(p)}$.
		\item[$\bullet$]$\epsilon=\sqrt{(-1)^{\frac{p-1}{2}}}$.
		\item[$\bullet$]$\zeta_n=e^{\frac{2\pi \sqrt{-1}}{n}}$ is the $n$-th complex primitive root of unity.
		\item[$\bullet$] $\delta_I$ is the indicator function. If $I=\{a\}$, then denote $\delta_{\{a\}}$ by $\delta_a$.
		\item [$\bullet$] For any set $D\subseteq V_n^{(p)}$, $\chi_a(D)$ is defined by $\chi_a(D)=\sum_{x\in D}\chi_a(x)$, where $a\in V_n^{(p)}$ and $\chi_a(x)=\zeta_p^{\langle a,x\rangle_n}$ is the additive character of $V_n^{(p)}$.
		\item[$\bullet$] $SQ$ and $NSQ$ denote the sets of squares and nonsquares in $\mathbb{F}_{p^m}^*$, respectively. 
		\item[$\bullet$]$\lambda_a$, $a\in\mathbb{F}_{p^m}$, denotes the additive character of $\mathbb{F}_{p^m}$.
		\item[$\bullet$]$\varphi_j$, $j=0,1,\dots,p^m-2$, denotes the multiplicative character of $\mathbb{F}_{p^m}$. In particular, $\eta_m$ is the quadratic multiplicative character of $\mathbb{F}_{p^m}$, i.e., $\eta_m(x)=1$ if $x\in SQ$ and $\eta_m(x)=-1$ if $x \in NSQ$.
		\item[$\bullet$]For a function $F:V_n^{(p)}\longrightarrow\mathbb{F}_{p^m}$, let $D_{F,I}=\{x\in V_n^{(p)}:F(x)\in I\}$, where $I\subseteq \mathbb{F}_{p^m}$. If $I={a}$, then denote $D_{F,\{a\}}$ by $D_{F,a}$.
		\item[$\bullet$] $\mathcal{E}_K=\{(1,0,\dots,0),(0,1,\dots,0),\dots,(0,\dots,0, 1)\}$ is the standard basis of the $K$-dimensional Hilbert space, where $K$ is a positive integer.
	\end{itemize}

	\subsection{Characters and character sums over finite fields}
	
	For any $a\in\mathbb{F}_{p^m}$, an \emph{additive character} $\lambda_a$ of $\mathbb{F}_{p^m}$ is defined as $\lambda_a(x)=\zeta_p^{\mathrm{Tr}_1^m(ax)}$. If $a=0$, the character $\lambda_0$ is called the \emph{trivial} additive character of $\mathbb{F}_{p^m}$. If $a=1$, the character $\lambda_1$ is called the \emph{canonical} additive character of $\mathbb{F}_{p^m}$. The \emph{conjugate character} $\overline{\lambda_a}$ of $\lambda_a$ is defined by $\overline{\lambda_a}(x)=\overline{\lambda_a(x)}=\lambda_{-a}(x)$ for $a, x\in \mathbb{F}_{p^m}$. 
	All additive characters of $\mathbb{F}_{p^m}$ form a group under the multiplication of additive characters. 
	
	
	Let $\mathbb{F}_{p^m}^*=\langle\alpha\rangle$, where $\alpha$ is a primitive element of $\mathbb{F}_{p^m}$. For any $j=0,1,\dots, p^m-2$, a \emph{multiplicative character} $\varphi_j$ of $\mathbb{F}_{p^m}$ is defined as 
	$\varphi_j(\alpha^i)=\zeta_{p^m-1}^{ij}$. The character $\varphi_0$ is called the \emph{trivial} multiplicative character of $\mathbb{F}_{p^m}.$ If $j=\frac{p^m-1}{2}$, then the character $\eta_m=\varphi_{\frac{p^m-1}{2}}$ is called the \emph{quadratic multiplicative character}.  Let $\overline{\varphi_j}$ be the \emph{conjugate character} of $\varphi_j$, which is defined by $\overline{\varphi_j}(x)=\overline{\varphi_j(x)}=\varphi_{j}(x^{-1})$ for $0\le j\le p^m-2$, $x\in\mathbb{F}_{p^m}^*$. 
	All multiplicative characters of $\mathbb{F}_{p^m}$ form a group under the multiplication of multiplicative characters.  
	Furthermore, we extend the definition of $\varphi_j$ to $\mathbb{F}_{p^m}$ by setting $\varphi_j(0)=1$ if $j=0$ and $\varphi_j(0)=0$ if $j\ne 0$.
	
	Let $\varphi$ be a multiplicative character and $\lambda$ be an additive character of $\mathbb{F}_{p^m}$, then the \emph{Gaussian sum} is defined by
	{\small \begin{equation*}
			G(\varphi,\lambda)=\sum\limits_{x\in\mathbb{F}_{p^m}^*}\varphi(x)\lambda(x).
	\end{equation*}}
	
	In the following lemmas, we introduce some results of Gaussian sums. 
	\begin{lemma}[\cite{Lidl}]\label{Le1}
		Let $\varphi$ be a multiplicative character and $\lambda$ be an additive character of $\mathbb{F}_{p^m}$, then
		{\small  \begin{equation*}
				G(\varphi,\lambda)=\begin{cases}
					p^m-1,\  &\text{if}\ \varphi=\varphi_0,\  \lambda=\lambda_0,\\
					-1,\ &\text{if}\ \varphi=\varphi_0,\ \lambda\ne\lambda_0,\\
					0,\ & \text{if}\ \varphi\ne\varphi_0,\ \lambda=\lambda_0.
				\end{cases}
		\end{equation*}}
		If $\varphi\ne\varphi_0$ and $\lambda\ne \lambda_0$, then $|G(\varphi,\lambda)|=\sqrt{p^m}$.
	\end{lemma}
	\begin{lemma}[\cite{Lidl}]\label{Le2}
		Gaussian sums for the finite field $\mathbb{F}_{p^m}$ satisfy $G(\varphi,\lambda_{ab})=\overline{\varphi}(a)G(\varphi,\lambda_b)$ for $a\in\mathbb{F}_{p^m}^*$, $b\in \mathbb{F}_{p^m}$.
	\end{lemma}
	Let $\varphi=\eta_m$, the associated Gaussian sum $G(\eta_m,\lambda)$ is then called the quadratic Gaussian sum for $\lambda\ne \lambda_0$. The following lemma gives its accurate value. 
	\begin{lemma}[\cite{Lidl}]\label{Le3}
		Let $\lambda_1$ be the canonical additive character of $\mathbb{F}_{p^m}$, then
		{\small 	\begin{align*}
				G(\eta_m,\lambda_1)&=(-1)^{m-1}\epsilon^m\sqrt{p^{m}}=\begin{cases}
					(-1)^{m-1}\sqrt{p^m},\ & \text{if}\ p\equiv 1\ (\mathrm{mod}\ 4),\\
					(-1)^{m-1}(\sqrt{-1})^m\sqrt{p^m}, \ &\text{if}\ p\equiv 3\ (\mathrm{mod}\ 4).
				\end{cases}
		\end{align*}}
	\end{lemma}
	
	Let $\varphi, \varphi'$ be multiplicative characters of $\mathbb{F}_{p^m}$, then the \emph{Jacobi sum} is defined by
	
	{\small \[J(\varphi,\varphi')=\sum\limits_{c_1+c_2=1,\ c_1, c_2\in\mathbb{F}_{p^m}}\varphi(c_1)\varphi'(c_2).\]}
	The values of the Jacobi sums are given in the following lemma.
	\begin{lemma}[\cite{Lidl}]\label{Le4}
		Let $\varphi, \varphi'$ be multiplicative characters of $\mathbb{F}_{p^m}$, then
		\begin{itemize}
			\item[($\mathrm{\romannumeral 1}$)]If $\varphi$ and $\varphi'$ are trivial, then $J(\varphi,\varphi')=p^m$.
			\item[($\mathrm{\romannumeral 2}$)] If one of $\varphi$ and $\varphi'$
			is trivial and the other is nontrivial, then  $J(\varphi,\varphi')=0$.
			\item [($\mathrm{\romannumeral 3}$)] If  $\varphi$ and $\varphi'$ are nontrivial, and $\varphi\varphi'$ is nontrivial, then  $J(\varphi,\varphi')=\frac{G(\varphi,\lambda)G(\varphi',\lambda)}{G(\varphi\varphi',\lambda)}$, where $\lambda$ is a nontrivial additive character of $\mathbb{F}_{p^m}$.
			\item [($\mathrm{\romannumeral 4}$)]  If  $\varphi$ and $\varphi'$ are nontrivial, and $\varphi\varphi'$ is trivial, then  $J(\varphi,\varphi')=-\frac{1}{p^m}G(\varphi,\lambda)G(\varphi',\lambda)$, where $\lambda$ is a nontrivial additive character of $\mathbb{F}_{p^m}$.
		\end{itemize} 
	\end{lemma}
	\subsection{Vectorial dual-bent functions}
	Let $f:V_n^{(p)}\longrightarrow\mathbb{F}_p$ be a $p$-ary function, the \emph{Walsh transform} of it is given by 
	{\small \[W_f(\alpha)=\sum\limits_{x\in V_n^{(p)}}\zeta_p^{f(x)-\langle \alpha,x\rangle_n}, \alpha\in V_n^{(p)}.\]} If for all $\alpha\in V_n^{(p)}$, $|W_f(\alpha)|=p^{\frac{n}{2}}$, then $f$ is called a $p$-ary \emph{bent} function. The Walsh transform of a bent function $f$ at $\alpha$ satisfies
	{\small 	\begin{equation*}
			W_f(\alpha)=\begin{cases}
				\pm \zeta_p^{f^*(\alpha)}p^{\frac{n}{2}}, \ &\text{if}\ p^n\equiv 1\ (\mathrm{mod}\ 4),\\
				\pm \sqrt{-1}\zeta_p^{f^*(\alpha)}p^{\frac{n}{2}}, \ &\text{if}\ p^n\equiv 3\ (\mathrm{mod}\ 4),
			\end{cases}
	\end{equation*}}
	where $f^*$ is a function from $V_n^{(p)}$ to $\mathbb{F}_p$, called the \emph{dual} of $f$
	\cite{Kumar}. A bent function $f$ is called \emph{weakly regular} if $W_f(\alpha)=\varepsilon_fp^{\frac{n}{2}}\zeta_p^{f^*(\alpha)}$ for all $\alpha\in V_n^{(p)}$, where $\varepsilon_f\in\{\pm1, \pm \sqrt{-1}\}$ is independent of $\alpha$. Otherwise, it is called \emph{non-weakly regular}. In particular, if $W_f(\alpha)=p^{\frac{n}{2}}\zeta_p^{f^*(\alpha)}$, then $f$ is said to be
	\emph{regular}.	It is known that the dual of a weakly regular bent function is also a weakly regular bent function, and satisfies $f^{**}(x)=f(-x)$, where $f^{**}$ is the dual of $f^*$. However, the dual of a non-weakly regular bent function is not necessarily a bent function\cite{Cesmelioglu1, Cesmelioglu2}.
	A vectorial $p$-ary function $F:V_n^{(p)}\longrightarrow V_m^{(p)}$ is called a \emph{vectorial bent} function, if every \emph{component function} $F_\alpha=\langle\alpha,F\rangle_m$, where $\alpha\in V_m^{(p)}\setminus\{0\}$, is a $p$-ary bent function. All component functions of a vectorial bent function together with the $0$-function form an $m$-dimensional vector space of bent functions over $\mathbb{F}_p$. Every $p$-ary  bent function is a vectorial bent function ($m=1$). A vectorial $p$-ary function $F$ from $V_n^{(p)}$ to $V_m^{(p)}$ is said to be \emph{vectorial dual-bent} if the duals of the component functions of $F$ together with the $0$-function also form an $m$-dimensional vector space of bent functions over $\mathbb{F}_p$. Equivalently, the bent functions $(F_\alpha)^*, \alpha\in V_m^{(p)}\setminus\{0\}$, are then the component functions of another vectorial bent function $F^*$ from $V_n^{(p)}$ to $V_{m}^{(p)}$, which is said to be a \emph{vectorial dual} of $F$. Obviously, for any $\alpha\in V_m^{(p)}\setminus\{0\}$, $(F_\alpha)^*=(F^*)_{\sigma(\alpha)}$, where $\sigma$ is some permutation of $V_m^{(p)}\setminus\{0\}$. 
	\section{Some results on vectorial dual-bent functions}
	In this section, we consider the set $\mathscr{F}$ of vectorial dual-bent functions $F:V_{n}^{(p)}\longrightarrow \mathbb{F}_{p^m}$ with $2\mid n$, $m\le \frac{n}{2}$, where $n,m$ are positive integers, satisfying the following conditions:
	\begin{itemize}
		\item [$\bullet$] There exists an integer $d$ with $\mathrm{gcd}(d-1,p^m-1)=1$ such that $(F_\alpha)^*=(F^*)_{\alpha^{1-d}}$ for $\alpha\in\mathbb{F}_{p^m}^*$.
		\item[$\bullet$] $F(-x)=F(x)$ for $x\in V_n^{(p)}$ and $F(0)=0$.
		\item[$\bullet$] All component functions $F_\alpha$, $\alpha\in \mathbb{F}_{p^m}^*$, are weakly regular and $\varepsilon_{F_\alpha}=\varepsilon$, where $\varepsilon=\{\pm1\}$ is a constant.
	\end{itemize}
	
	From the results in \cite{Wang1, Wang2}, we can get some classes of vectorial dual-bent functions belonging to $\mathscr{F}$.
	In the following lemmas, we give some results on vectorial dual-bent functions belonging to $\mathscr{F}$, which play a key role in constructing codebooks.
	\begin{lemma}[\cite{Wang3}]\label{Le5}
		Let $F:V_n^{(p)}\longrightarrow\mathbb{F}_{p^m}$ be a vectorial dual-bent function belonging to $\mathscr{F}$ and $F^*$ be its dual, then $F^*$ is also a vectorial dual-bent
		function with $F^*(-x)=F^*(x)$ and $F^*(0)=0$.
	\end{lemma}
	\begin{lemma}[\cite{Wang1}]\label{Le6}
		Let $F:V_n^{(p)}\longrightarrow\mathbb{F}_{p^m}$ be a vectorial dual-bent function belonging to $\mathscr{F}$ and $F^*$ be its dual, then
		\[|D_{F,i}|=|D_{F^*,i}|=p^{n-m}+\varepsilon p^{\frac{n}{2}-m}(p^m\delta_0(i)-1),\ \text{for\ any}\  i\in\mathbb{F}_{p^m}.\]
	\end{lemma}
	\begin{lemma}[\cite{Heng1, Wang1}]\label{Le7}
		Let $F:V_n^{(p)}\longrightarrow\mathbb{F}_{p^m}$ be a vectorial dual-bent function belonging to $\mathscr{F}$, then for any $a\in V_n^{(p)}\setminus\{0\}$,
		\begin{itemize}
			\item [($\mathrm{\romannumeral 1}$)]{\small $\chi_a(D_{F,0})=\begin{cases}
					\varepsilon p^{\frac{n}{2}-m}(p^m-1),\hspace{2.1cm}  \text{if}\ F^*(a)=0,\\
					-\varepsilon p^{\frac{n}{2}-m}, \hspace{3.15cm}\text{if}\ F^*(a)\ne0.
				\end{cases}$}
			\item [($\mathrm{\romannumeral 2}$)]{\small $\chi_a(D_{F,SQ})=\begin{cases}
					-\varepsilon p^{\frac{n}{2}-m}\frac{(p^m-1)}{2}, \hspace{1.85cm}\text{if}\ F^*(a)=0,\\
					\varepsilon p^{\frac{n}{2}-m}\frac{(\eta_m(F^*(a))p^m+1)}{2}, \hspace{0.8cm} \text{if}\ F^*(a)\ne0.
				\end{cases}$}
			\item [($\mathrm{\romannumeral 3}$)]{\small $\chi_a(D_{F,NSQ})=\begin{cases}
					-\varepsilon p^{\frac{n}{2}-m}\frac{(p^m-1)}{2},\ & \text{if}\ F^*(a)=0,\\
					\varepsilon p^{\frac{n}{2}-m}\frac{(-\eta_m(F^*(a))p^m+1)}{2}, \ & \text{if}\ F^*(a)\ne0.
				\end{cases}$}
		\end{itemize}
	\end{lemma}
	\begin{lemma}[\cite{Wang1}]\label{Le8}
		Let $F:V_n^{(p)}\longrightarrow\mathbb{F}_{p^m}$ be a vectorial dual-bent function belonging to $\mathscr{F}$ such that all component functions of $F$ are regular when $p>3$ and $(F_\alpha)^*=(F^*)_{\alpha}$, $\alpha\in \mathbb{F}_{p^m}^*$,
		then for any $a\in V_n^{(p)}\setminus\{0\}$, {\small \[\chi_a(D_{F,i})=\begin{cases}
				\varepsilon(p^m-1)p^{\frac{n}{2}-m}, \ & \text{if}\ F^*(a)=i,\\
				-	\varepsilon p^{\frac{n}{2}-m}, \ &\text{if}\ F^*(a)\ne i,
			\end{cases}\]}
		where $i\in\mathbb{F}_{p^m}^*$, $\varepsilon=\pm 1$ when $p=3$ and $\varepsilon=1$ when $p>3$.
	\end{lemma}
	\begin{remark}
		According to \cite[Proposition 1]{Wang3}, for the vectorial dual-bent function $F$ belonging to $\mathscr{F}$ with $(F_\alpha)^*=(F^*)_{\alpha}$, $\alpha\in\mathbb{F}_{p^m}^*$, if $p>3$, then all component functions of $F$ are regular. Thus, in Lemma \ref{Le8}, we let all component functions of $F$ be regular when $p>3$.
	\end{remark}
	\begin{lemma}[\cite{Heng1}]\label{Le9}
		Let $F:V_n^{(p)}\longrightarrow\mathbb{F}_{p^m}$ be a vectorial dual-bent function belonging to $\mathscr{F}$, $\varphi$ be a nontrivial multiplicative character of $\mathbb{F}_{p^m}$, and $\chi_a$ be the additive character of $V_{n}^{(p)}$, where $a\in V_n^{(p)}\setminus\{0\}$,  then
		{\small 	\[\sum\limits_{x\in V_n^{(p)}}\varphi(F(x))\chi_a(x)=\begin{cases}
				0,&\text{if}\ F^*(a)=0,\\
				\varepsilon p^{\frac{n}{2}-m}G(\varphi,\lambda_1)\varphi(-1)G(\overline{\rho},\lambda_1)\rho(F^*(a)), \ &\text{if}\ F^*(a)\ne 0,
			\end{cases}\]}
		where $\rho$ satisfies $\rho^{d-1}=\varphi^{-1}$.
	\end{lemma}
	\begin{remark}
	By the proofs of Theorems 2 and 3 in \cite{Heng1}, the condition that $F(ax)=a^lF(x)$ for  $a\in\mathbb{F}_{p^t}^*$ and $x\in V_n^{(p)}$, where $(l-1)(d-1)\equiv 1\ (\mathrm{mod}\ p^m-1), t|n,t|m$, is stronger than necessary. In fact, the result in Lemma \ref{Le9} also holds for any vectorial dual-bent function $F\in\mathscr{F}$.
	\end{remark}
	\begin{lemma}\label{Le10}
		Let $F:V_n^{(p)}\longrightarrow\mathbb{F}_{p^m}$ be a vectorial dual-bent function belonging to $\mathscr{F}$ such that all component functions of $F$ are regular when $p>3$ and $(F_\alpha)^*=(F^*)_{\alpha}$, $\alpha\in \mathbb{F}_{p^m}^*$. Let  $\varphi\ne \eta_m$ be a nontrivial multiplicative character of $\mathbb{F}_{p^m}$, and $\chi_a$ be the additive character of $V_{n}^{(p)}$, where $a\in V_n^{(p)}$, then
		\begin{itemize}
			\item [($\mathrm{\romannumeral 1}$)]
			{\small $\sum\limits_{x\in D_{F,SQ}}\varphi(F(x))\chi_a(x)=\begin{cases}
					0,\ & \text{if}\ a=0,\ \text{or}\ a\ne 0, \ F^*(a)\in NSQ\cup\{0\},\\
					\varepsilon p^{\frac{n}{2}}\varphi(F^*(a)),\ & \text{if}\ a\ne 0,\ \text{and}\  F^*(a)\in SQ,
				\end{cases}$}
			\item [($\mathrm{\romannumeral 2}$)]
			{\small 	$\sum\limits_{x\in D_{F,NSQ}}\varphi(F(x))\chi_a(x)=\begin{cases}
					0,\ & \text{if}\ a=0,\ \text{or}\ a\ne 0, \ F^*(a)\in SQ\cup\{0\},\\
					\varepsilon	p^{\frac{n}{2}}\varphi(F^*(a)),\ & \text{if}\ a\ne 0,\ \text{and}\  F^*(a)\in NSQ,
				\end{cases}$}
		\end{itemize}
		where $\varepsilon=\pm 1$ when $p=3$ and $\varepsilon=1$ when $p>3$.
	\end{lemma}
	\begin{proof}
		We only give the proof for $\sum\limits_{x\in D_{F,SQ}}\varphi(F(x))\chi_a(x)$, and the case of $\sum\limits_{x\in D_{F,NSQ}}\varphi(F(x))\chi_a(x)$ is similar.
		
		When $a=0$, let $\beta$ be a primitive element of $\mathbb{F}_{p^m}$, then we have that $\beta^2\in SQ$ and $\varphi(\beta^2)\ne 1$. 
		Note that $\sum\limits_{i\in SQ}\varphi(i)=\varphi(\beta^2)\sum\limits_{i\in SQ}\varphi(i)$, which implies that
		$\sum\limits_{i\in SQ}\varphi(i)=0$. Hence, $\sum\limits_{x\in D_{F,SQ}}\varphi(F(x))\chi_a(x)=\sum\limits_{i\in SQ}|D_{F,i}|\varphi(i)=(p^{n-m}-\varepsilon p^{\frac{n}{2}-m})\sum\limits_{i\in SQ}\varphi(i)=0$.
		
		When $a\ne 0$, by Lemma \ref{Le8}, we have that $\sum\limits_{x\in D_{F,SQ}}\varphi(F(x))\chi_a(x)=\sum\limits_{i\in SQ}\varphi (i)\chi_a(D_{F,i})=\varepsilon p^{\frac{n}{2}-m}\sum\limits_{i\in SQ}\varphi(i)(p^m\delta_i(F^*(a))-1).$
		If $F^*(a)\in NSQ\cup\{0\}$, then $\sum\limits_{x\in D_{F,SQ}}\varphi(F(x))\chi_a(x)=-\varepsilon p^{\frac{n}{2}-m}\sum\limits_{i\in SQ}\varphi(i)=0.$ If $F^*(a)\in SQ$, then $\sum\limits_{x\in D_{F,SQ}}\varphi(F(x))\chi_a(x)=-\varepsilon p^{\frac{n}{2}-m}\sum\limits_{i\in SQ}\varphi(i)+\varepsilon p^{\frac{n}{2}}\varphi(F^*(a))=\varepsilon p^{\frac{n}{2}}\varphi(F^*(a)).$
	\end{proof}
	\section{Constructions of codebooks based on additive characters}
	In this section, based on additive characters, we give two constructions of codebooks from vectorial dual-bent functions. By which, we obtain some asymptotically optimal codebooks with respect to the Welch bound.
	\subsection{The first construction of codebooks}
	Let $D$ be a subset of $V_{n}^{(p)}$ and $|D|=K$. In this subsection, we consider the codebook defined by
	{\small \begin{equation}\mathcal{C}_{D}=\{\mathbf{c}_{a}=\frac{1}{\sqrt{K}}(\chi_a(x))_{x\in D}:a\in V_{n}^{(p)}\}.
	\end{equation}}
	Selecting different sets $D$, we present several families of asymptotically optimal codebooks as follows.
	\begin{theorem}\label{Th1}Let $F:V_{n}^{(p)}\longrightarrow \mathbb{F}_{p^m}$ be a vectorial dual-bent function belonging to $\mathscr{F}$ such that all component functions of $F$ are regular when $p>3$ and $(F_{\alpha})^*=(F^*)_{\alpha}$, $\alpha\in\mathbb{F}_{p^m}^*$. Let $I$ be a nonempty subset of $\mathbb{F}_{p^m}$, $\frac{p^m-p^s}{2}\le |I|\le \frac{p^m+p^s}{2}$, where $s<m$ is a positive integer, and $D_1=D_{F,I}$. Then $\mathcal{C}_{D_1}$ defined by $(1)$ is a $(p^n,K)$ asymptotically optimal codebook with {\small \[I_{\mathrm{max}}(\mathcal{C}_{D_1})=\mathrm{max}\{\frac{p^{\frac{n}{2}}(p^m-|I|)}{p^mK}, \frac{ p^{\frac{n}{2}}|I|}{p^mK}\},\]}
		where $K=(p^{n-m}-\varepsilon p^{\frac{n}{2}-m})|I|+\varepsilon\delta_{I}(0)p^{\frac{n}{2}}$, $\varepsilon=\pm1$ when $p=3$ and $\varepsilon=1$ when $p>3$. Furthermore, the distribution of cross-correlation amplitudes of $\mathcal{C}_{D_1}$ is given by
		{\small 	\[\lvert\mathbf{c}_{a_1}\mathbf{c}_{a_2}^{H}\rvert=\begin{cases}
				\frac{ p^{\frac{n}{2}}(p^m-|I|)}{p^mK},\ p^n((p^{n-m}-\varepsilon p^{\frac{n}{2}-m})|I|+\delta_{I}(0)(\varepsilon p^{\frac{n}{2}}-1))\ \text{times},\\
				\frac{ p^{\frac{n}{2}}|I|}{p^mK}, \ p^n(p^n-(p^{n-m}-\varepsilon p^{\frac{n}{2}-m})|I|-\delta_{I}(0)(\varepsilon p^{\frac{n}{2}}-1)-1)\ \text{times},
			\end{cases}\]}
		where $\mathbf{c}_{a_1},\mathbf{c}_{a_2}\in\mathcal{C}_{D_1}$ with $a_1\ne a_2$.
	\end{theorem}
	\begin{proof}
		By Lemma \ref{Le6}, we get that $K=(p^{n-m}-\varepsilon p^{\frac{n}{2}-m})|I|+	\delta_{I}(0)\varepsilon p^{\frac{n}{2}}$. Let $\mathbf{c}_{a_1}$ and $\mathbf{c}_{a_2}$ be two codewords in $\mathcal{C}_{D_1}$ with $a_1\ne a_2$, by Lemma \ref{Le8}, we have that 
		{\small \begin{align*}
				\mathbf{c}_{a_1}\mathbf{c}_{a_2}^H
				=\frac{1}{K}\sum\limits_{i\in I}\sum\limits_{x\in D_{F,i}}\chi_{a_1-a_2}(x)=\frac{1}{K}\sum\limits_{i\in I}\varepsilon(p^m\delta_i(F^*(a_1-a_2))-1)p^{\frac{n}{2}-m}.
		\end{align*}}
		If $F^*(a_1-a_2)\in I$, then $	\mathbf{c}_{a_1}\mathbf{c}_{a_2}^H=\frac{1}{K}\varepsilon p^{\frac{n}{2}-m}(p^m-|I|)$.
		If $F^*(a_1-a_2)\notin I$, then $	\mathbf{c}_{a_1}\mathbf{c}_{a_2}^H=-\frac{1}{K}\varepsilon p^{\frac{n}{2}-m}|I|$. Hence, we get that $I_{\mathrm{max}}(\mathcal{C}_1)=\mathrm{max}\{\frac{ p^{\frac{n}{2}}(p^m-|I|)}{p^mK}, \frac{ p^{\frac{n}{2}}|I|}{p^mK}\}$. Obviously, $\mathbf{c}_{a_1}\ne \mathbf{c}_{a_2}$, so $|\mathcal{C}_{D_1}|=p^n$. According to Lemma \ref{Le6}, we can easily obtain the  distribution of cross-correlation amplitudes of $\mathcal{C}_{D_1}$. Now we discuss the asymptotic optimality of the codebook $\mathcal{C}_{D_1}$ in four cases.
		\begin{itemize}
			\item [$\bullet$]
			$\frac{p^m-p^s}{2}\le |I|\le \frac{p^m-1}{2}$ and $0\notin I$
			
			In this case, $K=(p^{n-m}-\varepsilon p^{\frac{n}{2}-m})|I|$, and $I_{\mathrm{max}}(\mathcal{C}_{D_1})=\frac{p^{\frac{n}{2}}(p^m-|I|)}{p^mK}$, then 
			{\small 			\begin{align*}
					\frac{I_{\mathrm{max}}(\mathcal{C}_{D_1})}{I_\mathrm{W}}
					&=\frac{p^{\frac{n}{2}}(p^m-|I|)}{p^m}\sqrt{\frac{p^n-1}{(p^{n-m}-\varepsilon p^{\frac{n}{2}-m})|I|(p^n-(p^{n-m}-\varepsilon p^{\frac{n}{2}-m})|I|)}}\\
					&\le \frac{p^{\frac{n}{2}}(p^m+p^s)}{2p^m}\sqrt{\frac{p^n-1}{(p^{n-m}-\varepsilon p^{\frac{n}{2}-m})|I|(p^n-(p^{n-m}-\varepsilon p^{\frac{n}{2}-m})|I|)}}.
			\end{align*}}
			Since $\frac{p^n(p^{n-m}-\varepsilon p^{\frac{n}{2}-m})}{2(p^{n-m}-\varepsilon p^{\frac{n}{2}-m})^2}=\frac{p^n}{2(p^{n-m}-\varepsilon p^{\frac{n}{2}-m})}>\frac{p^m-1}{2}$, then
			{\small 	\begin{align*}
					\frac{I_{\mathrm{max}}(\mathcal{C}_{D_1})}{I_\mathrm{W}}&\le \sqrt{\frac{p^n(p^m+p^s)^2(p^n-1)}{4p^{2m}(p^{n-m}-\varepsilon p^{\frac{n}{2}-m})\frac{p^m-p^s}{2}(p^n-(p^{n-m}-\varepsilon p^{\frac{n}{2}-m})\frac{p^m-p^s}{2})}}\\
					&=\sqrt{\frac{p^n(p^m+p^s)^2(p^n-1)}{(p^{n+m}+p^{n+s}+\varepsilon p^{\frac{n}{2}+m}-\varepsilon p^{\frac{n}{2}+s})(p^n-\varepsilon p^{\frac{n}{2}})(p^m-p^s)}}.
			\end{align*}}
			Since $m\le \frac{n}{2}$, then  $\frac{I_{\mathrm{max}}(\mathcal{C}_{D_1})}{I_\mathrm{W}}\rightarrow 1$ if $p\rightarrow +\infty$ or $m-s\rightarrow +\infty$, which implies that the codebook $\mathcal{C}_{D_1}$ asymptotically achieves the Welch bound.
			\item[$\bullet$] $\frac{p^m-p^s}{2}\le |I|\le \frac{p^m-1}{2}$ and $0\in I$
			
			In this case, $K=(p^{n-m}-\varepsilon p^{\frac{n}{2}-m})|I|+\varepsilon p^{\frac{n}{2}}$, and $I_{\mathrm{max}}(\mathcal{C}_{D_1})=\frac{p^{\frac{n}{2}}(p^m-|I|)}{p^mK}$, then 
			{\small 	\begin{align*}
					\frac{I_{\mathrm{max}}(\mathcal{C}_{D_1})}{I_\mathrm{W}}
					&=\frac{p^{\frac{n}{2}}(p^m-|I|)}{p^m}\sqrt{\frac{p^n-1}{((p^{n-m}-\varepsilon p^{\frac{n}{2}-m})|I|+\varepsilon p^{\frac{n}{2}}) (p^n-(p^{n-m}-\varepsilon p^{\frac{n}{2}-m})|I|-\varepsilon p^{\frac{n}{2}})}}\\
					&\le \frac{p^{\frac{n}{2}}(p^m+p^s)}{2p^m}\sqrt{\frac{p^n-1}{((p^{n-m}-\varepsilon p^{\frac{n}{2}-m})|I|+\varepsilon p^{\frac{n}{2}}) (p^n-(p^{n-m}-\varepsilon p^{\frac{n}{2}-m})|I|-\varepsilon  p^{\frac{n}{2}})}}.
			\end{align*}}
			Note that  $\frac{(p^n-2\varepsilon p^{\frac{n}{2}})(p^{n-m}-\varepsilon p^{\frac{n}{2}-m})}{2(p^{n-m}-\varepsilon p^{\frac{n}{2}-m})^2}=\frac{p^n-2\varepsilon p^{\frac{n}{2}}}{2(p^{n-m}-\varepsilon p^{\frac{n}{2}-m})}>\frac{p^m-p^s}{2}$ and $|\frac{p^n-2\varepsilon p^{\frac{n}{2}}}{2(p^{n-m}-\varepsilon p^{\frac{n}{2}-m})}-\frac{p^m-p^s}{2}|>|\frac{p^n-2\varepsilon p^{\frac{n}{2}}}{2(p^{n-m}-\varepsilon p^{\frac{n}{2}-m})}-\frac{p^m-1}{2}|$, then
			{\small \begin{align*}
					\frac{I_{\mathrm{max}}(\mathcal{C}_{D_1})}{I_\mathrm{W}}&\le \sqrt{\frac{p^n(p^m+p^s)^2(p^n-1)}{4p^{2m}((p^{n-m}-\varepsilon p^{\frac{n}{2}-m})\frac{p^m-p^s}{2}+\varepsilon p^{\frac{n}{2}}) (p^n-(p^{n-m}-\varepsilon p^{\frac{n}{2}-m})\frac{p^m-p^s}{2}-\varepsilon p^{\frac{n}{2}})}}\\
					&=\sqrt{\frac{p^n(p^m+p^s)(p^n-1)}{(p^{n+m}+\varepsilon p^{\frac{n}{2}+m}+\varepsilon p^{\frac{n}{2}+s}-p^{n+s})(p^n-\varepsilon p^{\frac{n}{2}})}}.
			\end{align*}}
			Since $ m\le \frac{n}{2}$, then  $\frac{I_{\mathrm{max}}(\mathcal{C}_{D_1})}{I_\mathrm{W}}\rightarrow 1$ if $p\rightarrow +\infty$ or $m-s\rightarrow +\infty$,  which implies that the codebook $\mathcal{C}_{D_1}$ asymptotically achieves the Welch bound.
			\item[$\bullet$] $\frac{p^m+1}{2}\le |I|\le \frac{p^m+p^s}{2}$ and $0\notin I$
			
			In this case, $K=(p^{n-m}-\varepsilon p^{\frac{n}{2}-m})|I|$, and  $I_{\mathrm{max}}(\mathcal{C}_{D_1})=\frac{p^{\frac{n}{2}}|I|}{p^mK}$, then
			{\small \begin{align*}
					\frac{I_{\mathrm{max}}(\mathcal{C}_{D_1})}{I_\mathrm{W}}
					&=\frac{p^{\frac{n}{2}}|I|}{p^m}\sqrt{\frac{p^n-1}{(p^{n-m}-\varepsilon p^{\frac{n}{2}-m})|I|(p^n-(p^{n-m}-\varepsilon p^{\frac{n}{2}-m})|I|)}}\\
					&\le \frac{p^{\frac{n}{2}}(p^m+p^s)}{2p^m}\sqrt{\frac{p^n-1}{(p^{n-m}-\varepsilon p^{\frac{n}{2}-m})|I|(p^n-(p^{n-m}-\varepsilon p^{\frac{n}{2}-m})|I|)}}.
			\end{align*}}
			Since  $\frac{p^n(p^{n-m}-\varepsilon p^{\frac{n}{2}-m})}{2(p^{n-m}-\varepsilon p^{\frac{n}{2}-m})^2}=\frac{p^n}{2(p^{n-m}-\varepsilon p^{\frac{n}{2}-m})}<\frac{p^m+p^s}{2}$,  $|\frac{p^n}{2(p^{n-m}-\varepsilon p^{\frac{n}{2}-m})}-\frac{p^m+1}{2}|<|\frac{p^n}{2(p^{n-m}-\varepsilon p^{\frac{n}{2}-m})}\\-\frac{p^m+p^s}{2}|$, then
			{\small 	\begin{align*}
					\frac{I_{\mathrm{max}}(\mathcal{C}_{D_1})}{I_\mathrm{W}}&\le \sqrt{\frac{p^n(p^m+p^s)^2(p^n-1)}{4p^{2m}(p^{n-m}-\varepsilon p^{\frac{n}{2}-m})\frac{p^m+p^s}{2}(p^n-(p^{n-m}-\varepsilon p^{\frac{n}{2}-m})\frac{p^m+p^s}{2})}}\\
					&=\sqrt{\frac{p^n(p^m+p^s)(p^n-1)}{(p^{n+m}-p^{n+s}+\varepsilon p^{\frac{n}{2}+m}+\varepsilon p^{\frac{n}{2}+s})(p^n-\varepsilon p^{\frac{n}{2}})}}.
			\end{align*}}
			Since $m\le \frac{n}{2}$, then  $\frac{I_{\mathrm{max}}(\mathcal{C}_{D_1})}{I_\mathrm{W}}\rightarrow 1$ if $p\rightarrow +\infty$ or $m-s\rightarrow +\infty$,  which implies that the codebook $\mathcal{C}_{D_1}$ asymptotically achieves the Welch bound.
			\item[$\bullet$] $\frac{p^m+1}{2}\le |I|\le \frac{p^m+p^s}{2}$ and $0\in I$
			
			In this case, $K=(p^{n-m}-\varepsilon p^{\frac{n}{2}-m})|I|+\varepsilon p^{\frac{n}{2}}$, and  $I_{\mathrm{max}}(\mathcal{C}_{D_1})=\frac{p^{\frac{n}{2}}|I|}{p^mK}$, then
			{\small 	\begin{align*}
					\frac{I_{\mathrm{max}}(\mathcal{C}_{D_1})}{I_\mathrm{W}}
					&=\frac{p^{\frac{n}{2}}|I|}{p^m}\sqrt{\frac{p^n-1}{((p^{n-m}-\varepsilon p^{\frac{n}{2}-m})|I|+\varepsilon p^{\frac{n}{2}}) (p^n-(p^{n-m}-\varepsilon p^{\frac{n}{2}-m})|I|-\varepsilon p^{\frac{n}{2}})}}\\
					&\le \frac{p^{\frac{n}{2}}(p^m+p^s)}{2p^m}\sqrt{\frac{p^n-1}{((p^{n-m}-\varepsilon p^{\frac{n}{2}-m})|I|+\varepsilon p^{\frac{n}{2}}) (p^n-(p^{n-m}-\varepsilon p^{\frac{n}{2}-m})|I|-\varepsilon p^{\frac{n}{2}})}}.
			\end{align*}}
			Since $\frac{(p^n-2\varepsilon p^{\frac{n}{2}})(p^{n-m}-\varepsilon p^{\frac{n}{2}-m})}{2(p^{n-m}-\varepsilon p^{\frac{n}{2}-m})^2}=\frac{p^n-2\varepsilon p^{\frac{n}{2}}}{2(p^{n-m}-\varepsilon p^{\frac{n}{2}-m})}<\frac{p^m+1}{2}$, then {\small \begin{align*}
					\frac{I_{\mathrm{max}}(\mathcal{C}_{D_1})}{I_\mathrm{W}}&\le \sqrt{\frac{p^n(p^m+p^s)^2(p^n-1)}{4p^{2m}((p^{n-m}-\varepsilon p^{\frac{n}{2}-m})\frac{p^m+p^s}{2}+\varepsilon p^{\frac{n}{2}}) (p^n-(p^{n-m}-\varepsilon p^{\frac{n}{2}-m})\frac{p^m+p^s}{2}-\varepsilon p^{\frac{n}{2}})}}\\
					&=\sqrt{\frac{p^n(p^m+p^s)^2(p^n-1)}{(p^{n+m}+p^{n+s}+\varepsilon p^{\frac{n}{2}+m}-\varepsilon p^{\frac{n}{2}+s})(p^n-\varepsilon p^{\frac{n}{2}})(p^m-p^s)}}.
			\end{align*}}
			Since $m\le \frac{n}{2}$, then $\frac{I_{\mathrm{max}}(\mathcal{C}_{D_1})}{I_\mathrm{W}}\rightarrow 1$ if $p\rightarrow +\infty$ or $m-s\rightarrow +\infty$,  which implies that the codebook $\mathcal{C}_{D_1}$ asymptotically achieves the Welch bound.
		\end{itemize}
	\end{proof}
	\begin{remark}
		According to Theorem \ref{Th1}, $\frac{I_{\mathrm{max}}(\mathcal{C}_{D_1})}{I_\mathrm{W}}\rightarrow 1$ if $p\rightarrow +\infty$ or $m-s\rightarrow +\infty$. If  $p$ is very small and $m-s\rightarrow +\infty$, then the codebook $\mathcal{C}_{D_1}$ also asymptotically achieves the Welch bound, while its alphabet size is very small. In \cite[Theorem 6]{Heng1}, setting $D=D_{F,NSQ}$, Heng et al. constructed a family of $(p^n,\frac{(p^{n-m}-\varepsilon p^{\frac{n}{2}-m})(p^m-1)}{2})$ codebooks from the vectorial dual-bent function $F$, where $\varepsilon=\pm 1$. When $\varepsilon=1$, by Theorem \ref{Th1}, let $|I|=\frac{p^m-1}{2}$ and $0\notin I$, then the codebook $\mathcal{C}_{D_1}$ has the same parameters as those in \cite{Heng1}. However, in our construction, the choice of the set 
		$D$ is more diverse, and the parameters of the resulting codebooks are more flexible.
	\end{remark}
	\begin{theorem}\label{Th2}
		Let $F:V_{r}^{(p)}\longrightarrow\mathbb{F}_{p^m}$ be a vectorial dual-bent function belonging to $\mathscr{F}$ such that all component functions of $F$ are weakly regular but not regular, where $r=2m$. Let  $V_n^{(p)}=V_{r}^{(p)}\times\mathbb{F}_{p^m}$. 
		
		\begin{itemize}
			\item  [($\mathrm{\romannumeral 1}$)] Let $D_2=(D_{F,SQ}\times SQ)\cup((D_{F,NSQ}\cup D_{F,0})\times (NSQ\cup\{0\}))$, then $\mathcal{C}_{D_2}$ defined by $(1)$ is a $(p^{3m},K_1)$ asymptotically optimal codebook, and 
			{\small 	\[I_{\mathrm{max}}(\mathcal{C}_{D_2})=\begin{cases}
					\frac{(p^m+1)(p^{\frac{m}{2}}+1)}{2K_1}, \ \text{if}\ p^m\equiv 1 \ (\mathrm{mod}\ 4),\\
					\frac{(p^m+1)\sqrt{p^m+1}}{2K_1}, \ \text{if}\ p^m\equiv 3 \ (\mathrm{mod}\ 4),
				\end{cases}\]} where $K_1=\frac{p^m+1}{2}(p^{2m}-p^m+1)$. Furthermore, for two codewords $\mathbf{c}_{(a_1,a_2)}$, $\mathbf{c}_{(b_1,b_2)}\in \mathcal{C}_{D_2}$ with $(a_1,a_2)\ne (b_1,b_2)$, the distribution of their cross-correlation amplitude is given as follows.
			
			\begin{itemize}
				\item [$\bullet$] If $p^m\equiv 1 \ (\mathrm{mod}\ 4)$, then
				{\small 	\[|\mathbf{c}_{(a_1,a_2)}\mathbf{c}_{(b_1,b_2)}^H|=\begin{cases}
						\frac{p^{\frac{m}{2}}+1}{2K_1}, \ \frac{p^m-1}{2}p^{3m}\ \text{times},\\
							\frac{p^{\frac{m}{2}}-1}{2K_1}, \ \frac{p^m-1}{2}p^{3m}\ \text{times},\\
						\frac{p^m+1}{2K_1},\ \frac{p^{2m}-1}{2}p^{3m}\ \text{times},\\
						\frac{p^m-1}{2K_1},\ \frac{p^{2m}-1}{2}p^{3m}\ \text{times},\\
							\frac{(p^m+1)(p^{\frac{m}{2}}+1)}{2K_1}, \ \frac{(p^{2m}-1)(p^m-1)}{4}p^{3m}\ \text{times},\\
							\frac{(p^m+1)(p^{\frac{m}{2}}-1)}{2K_1}, \ \frac{(p^{2m}-1)(p^m-1)}{4}p^{3m}\ \text{times},\\
							\frac{(p^m-1)(p^{\frac{m}{2}}+1)}{2K_1}, \ \frac{(p^{2m}-1)(p^m-1)}{4}p^{3m}\ \text{times},\\
						\frac{(p^m-1)(p^{\frac{m}{2}}-1)}{2K_1}, \ \frac{(p^{2m}-1)(p^m-1)}{4}p^{3m}\ \text{times}.
					\end{cases}\]}
				\item [$\bullet$] If $p^m\equiv 3 \ (\mathrm{mod}\ 4)$, then
				{\small 	\[|\mathbf{c}_{(a_1,a_2)}\mathbf{c}_{(b_1,b_2)}^H|=\begin{cases}
						\frac{\sqrt{p^{m}+1}}{2K_1}, \ (p^m-1)p^{3m}\ \text{times},\\
						\frac{p^m+1}{2K_1},\ \frac{p^{2m}-1}{2}p^{3m}\ \text{times},\\
						\frac{p^m-1}{2K_1},\ \frac{p^{2m}-1}{2}p^{3m}\ \text{times},\\
							\frac{(p^m+1)\sqrt{p^{m}+1}}{2K_1}, \ \frac{(p^{2m}-1)(p^m-1)}{2}p^{3m}\ \text{times},\\
						\frac{(p^m-1)\sqrt{p^{m}+1}}{2K_1}, \ \frac{(p^{2m}-1)(p^m-1)}{2}p^{3m}\ \text{times}.
					\end{cases}\]}
			\end{itemize}
			\item  [($\mathrm{\romannumeral 2}$)] Let $D_3=(D_{F,SQ}\times (NSQ\cup\{0\}))\cup((D_{F,NSQ}\cup D_{F,0})\times SQ)$, then $\mathcal{C}_{D_3}$ defined by $(1)$ is a $(p^{3m},K_2)$ asymptotically optimal codebook, and 
			{\small 	\[I_{\mathrm{max}}(\mathcal{C}_{D_3})=\begin{cases}
					\frac{(p^m+1)(p^{\frac{m}{2}}+1)}{2K_2}, \ \text{if}\ p^m\equiv 1 \ (\mathrm{mod}\ 4),\\
					\frac{(p^m+1)\sqrt{p^m+1}}{2K_2}, \ \text{if}\ p^m\equiv 3 \ (\mathrm{mod}\ 4),
				\end{cases}\]} where $K_2=\frac{p^m-1}{2}(p^{2m}+p^m+1)$. Furthermore, for two codewords $\mathbf{c}_{(a_1,a_2)}$, $\mathbf{c}_{(b_1,b_2)}\in \mathcal{C}_{D_3}$ with $(a_1,a_2)\ne (b_1,b_2)$, the distribution of their cross-correlation amplitude is given as follows.
			\begin{itemize}
				\item [$\bullet$] If $p^m\equiv 1 \ (\mathrm{mod}\ 4)$, then
				
				{\small 	\[|\mathbf{c}_{(a_1,a_2)}\mathbf{c}_{(b_1,b_2)}^H|=\begin{cases}
					\frac{p^{\frac{m}{2}}+1}{2K_2}, \ \frac{p^m-1}{2}p^{3m}\ \text{times},\\
					\frac{p^{\frac{m}{2}}-1}{2K_2}, \ \frac{p^m-1}{2}p^{3m}\ \text{times},\\
					\frac{p^m+1}{2K_2},\ \frac{p^{2m}-1}{2}p^{3m}\ \text{times},\\
					\frac{p^m-1}{2K_2},\ \frac{p^{2m}-1}{2}p^{3m}\ \text{times},\\
					\frac{(p^m+1)(p^{\frac{m}{2}}+1)}{2K_2}, \ \frac{(p^{2m}-1)(p^m-1)}{4}p^{3m}\ \text{times},\\
					\frac{(p^m+1)(p^{\frac{m}{2}}-1)}{2K_2}, \ \frac{(p^{2m}-1)(p^m-1)}{4}p^{3m}\ \text{times},\\
					\frac{(p^m-1)(p^{\frac{m}{2}}+1)}{2K_2}, \ \frac{(p^{2m}-1)(p^m-1)}{4}p^{3m}\ \text{times},\\
					\frac{(p^m-1)(p^{\frac{m}{2}}-1)}{2K_2}, \ \frac{(p^{2m}-1)(p^m-1)}{4}p^{3m}\ \text{times}.
				\end{cases}\]}
				\item [$\bullet$] If $p^m\equiv 3 \ (\mathrm{mod}\ 4)$, then
				{\small 	\[|\mathbf{c}_{(a_1,a_2)}\mathbf{c}_{(b_1,b_2)}^H|=\begin{cases}
					\frac{\sqrt{p^{m}+1}}{2K_2}, \ (p^m-1)p^{3m}\ \text{times},\\
					\frac{p^m+1}{2K_2},\ \frac{p^{2m}-1}{2}p^{3m}\ \text{times},\\
					\frac{p^m-1}{2K_2},\ \frac{p^{2m}-1}{2}p^{3m}\ \text{times},\\
					\frac{(p^m+1)\sqrt{p^{m}+1}}{2K_2}, \ \frac{(p^{2m}-1)(p^m-1)}{2}p^{3m}\ \text{times},\\
					\frac{(p^m-1)\sqrt{p^{m}+1}}{2K_2}, \ \frac{(p^{2m}-1)(p^m-1)}{2}p^{3m}\ \text{times}.
				\end{cases}\]}
				
			\end{itemize}
		\end{itemize}
		
	\end{theorem}
	\begin{proof}
		We only give the proof of ($\mathrm{\romannumeral 1}$), and the proof of ($\mathrm{\romannumeral 2}$) is similar. 
		
		According to Lemmas \ref{Le1}, \ref{Le2} and \ref{Le3}, for any $a\in\mathbb{F}_{p^m}^*$, we have that $\lambda_a(SQ)-\lambda_a(NSQ)=\eta_m(a)G(\eta_m,\lambda_1)$ and $\lambda_a(SQ)+\lambda_a(NSQ)=-1$. Thus, $\lambda_a(SQ)=\frac{\eta_m(a)G(\eta_m,\lambda_1)-1}{2}$ and $\lambda_a(NSQ)=\frac{-\eta_m(a)G(\eta_m,\lambda_1)-1}{2}$.
		
		By Lemma \ref{Le6}, we easily get that $K_1=\frac{p^m+1}{2}(p^{2m}-p^m+1)$. Let $\mathbf{c}_{(a_1,a_2)}$, $\mathbf{c}_{(b_1,b_2)}$ be two codewords in $\mathcal{C}_{D_2}$ with $(a_1,a_2)\ne (b_1,b_2)$, and let $\mu_a$ denote the additive character of $V_r^{(p)}$, defined by $\mu_a(x)=\zeta_p^{\langle a,x\rangle_r}$, $a,x\in V_r^{(p)}$. Then 
		{\small 	\begin{align*}
				\mathbf{c}_{(a_1,a_2)}\mathbf{c}_{(b_1,b_2)}^H=\frac{1}{K_1}\left(\mu_{a_1-b_1}(D_{F,SQ})\lambda_{a_2-b_2}(SQ)+\mu_{a_1-b_1}(D_{F,NSQ}\cup D_{F,0})\lambda_{a_2-b_2}(NSQ\cup\{0\})\right).
		\end{align*}}
		According to Lemmas \ref{Le6} and \ref{Le7}, we calculate the value of $\mathbf{c}_{(a_1,a_2)}\mathbf{c}_{(b_1,b_2)}^H$ by considering the following cases.
		\begin{itemize}
			\item [$\bullet$] If $a_1=b_1$ and $a_2\ne b_2$, then
			{\small 	\begin{align*}
					\mathbf{c}_{(a_1,a_2)}\mathbf{c}_{(b_1,b_2)}^H&=\frac{1}{K_1}(|D_{F,SQ}|\frac{G(\eta_m,\lambda_1)\eta_m(a_2-b_2)-1}{2}-(|D_{F,NSQ}|+|D_{F,0}|)\frac{G(\eta_m,\lambda_1)\eta_m(a_2-b_2)-1}{2})\\
					&=-\frac{G(\eta_m,\lambda_1)\eta_m(a_2-b_2)-1}{2K_1}.
			\end{align*}}
			\item [$\bullet$] If $a_1\ne b_1$ and $a_2=b_2$, then $F^*(a_1-b_1)
			\ne 0$, thus		
			{\small \begin{align*}
					\mathbf{c}_{(a_1,a_2)}\mathbf{c}_{(b_1,b_2)}^H
					&=\frac{1}{K_1}(-\frac{(\eta_m(F^*(a_1-b_1))p^m+1)(p^m-1)}{4}+\frac{(\eta_m(F^*(a_1-b_1))p^m+1)(p^m+1)}{4})\\
					&=\frac{\eta_m(F^*(a_1-b_1))p^m+1}{2K_1}.
			\end{align*}}
			\item [$\bullet$] If $a_1\ne b_1$ and $a_2\ne b_2$, then $F^*(a_1-b_1)
			\ne 0$, thus
			{\small \begin{align*}
					\mathbf{c}_{(a_1,a_2)}\mathbf{c}_{(b_1,b_2)}^H
					&=-\frac{(\eta_m(F^*(a_1-b_1))p^m+1)(G(\eta_m,\lambda_1)\eta_m(a_2-b_2)-1)}{4K_1}\\&+\frac{(\eta_m(F^*(a_1-b_1))p^m+1)(-G(\eta_m,\lambda_1)\eta_m(a_2-b_2)+1)}{4K_1}\\
					&=-\frac{(\eta_m(F^*(a_1-b_1))p^m+1)(G(\eta_m,\lambda_1)\eta_m(a_2-b_2)-1)}{2K_1}.
			\end{align*}}
		\end{itemize}
		Hence, by Lemma \ref{Le3},
		we have that if $p^m\equiv 1 \ (\mathrm{mod}\ 4)$, then $I_{\mathrm{max}}(\mathcal{C}_{D_2})=\frac{(p^m+1)(p^{\frac{m}{2}}+1)}{2K_1}$; if $p^m\equiv 3 \ (\mathrm{mod}\ 4)$, 
		then
		$I_{\mathrm{max}}(\mathcal{C}_{D_2})=\frac{(p^m+1)\sqrt{p^m+1}}{2K_1}$. Obviously, $\mathbf{c}_{(a_1,a_2)}\ne \mathbf{c}_{(b_1,b_2)}$, so $|\mathcal{C}_{D_2}|=p^{3m}$.
		When $p^m\equiv 1 \ (\mathrm{mod}\ 4)$, then 
		{\small \begin{align*}
				\frac{I_{\mathrm{max}}(\mathcal{C}_{D_2})}{I_{\mathrm{W}}}
				=\sqrt{\frac{(p^m+1)(p^{\frac{m}{2}}+1)^2}{p^{2m}-p^m+1}}.
		\end{align*}}
		Therefore, $\frac{I_{\mathrm{max}}(\mathcal{C}_{D_2})}{I_{\mathrm{W}}}\rightarrow 1$ if $p\rightarrow +\infty$ or $m\rightarrow +\infty$,  which implies that the codebook $\mathcal{C}_{D_2}$ asymptotically achieves the Welch bound.
		
		When $p^m\equiv 3 \ (\mathrm{mod}\ 4)$, then 
		{\small \begin{align*}
				\frac{I_{\mathrm{max}}(\mathcal{C}_{D_2})}{I_{\mathrm{W}}}
				=\sqrt{\frac{(p^m+1)^2}{p^{2m}-p^m+1}}.
		\end{align*}}
		Therefore, $\frac{I_{\mathrm{max}}(\mathcal{C}_{D_2})}{I_{\mathrm{W}}}\rightarrow 1$ if $p\rightarrow +\infty$ or $m\rightarrow +\infty$, which implies that the codebook $\mathcal{C}_{D_2}$ asymptotically achieves the Welch bound.
		
		By Lemma \ref{Le6} and the above discussion, we can easily obtain the distribution
		of the cross-correlation amplitudes of $\mathcal{C}_{D_2}$.  
	\end{proof}	
	\begin{theorem}\label{Th3}
		Let $F$ and $H$ be vectorial dual-bent functions from $V_r^{(p)}$ to $\mathbb{F}_{p^m}$ that belong to $\mathscr{F}$, where $r=2m$, and all component functions of $F$ and $H$ be weakly regular but not regular. Let $V_n^{(p)}=V_r^{(p)}\times V_r^{(p)}$.
		\begin{itemize}
			\item  [($\mathrm{\romannumeral 1}$)] Let $D_4=(D_{F,SQ}\times D_{H,SQ})\cup((D_{F,NSQ}\cup D_{F,0})\times(D_{H,NSQ}\cup D_{H,0}))$, then $\mathcal{C}_{D_4}$ defined by $(1)$ is a $(p^{4m}, K_1)$ asymptotically optimal codebook with 
			{\small \[I_{\mathrm{max}}(\mathcal{C}_{D_4})=\frac{(p^m+1)^2}{2K_1},\] }
			where $K_1=\frac{p^{4m}+1}{2}$. Furthermore, the distribution of cross-correlation amplitudes of $\mathcal{C}_{D_4}$ is given by 
			{\small \[|\mathbf{c}_{(a_1,a_2)}\mathbf{c}_{(b_1,b_2)}^H|=\begin{cases}
 					\frac{p^m+ 1}{2K_1},\  (p^{2m}-1)p^{4m}\ \text{times},\\	
					\frac{p^m- 1}{2K_1},\  (p^{2m}-1)p^{4m}\ \text{times},\\
					\frac{(p^m+ 1)^2}{2K_1}, \ \frac{(p^{2m}-1)^2}{4}p^{4m}\ \text{times},\\
					\frac{(p^m-1)^2}{2K_1}, \ \frac{(p^{2m}-1)^2}{4}p^{4m}\ \text{times},\\
					\frac{p^{2m}- 1}{2K_1}, \ \frac{(p^{2m}-1)^2}{2}p^{4m}\ \text{times},
				\end{cases}\]}
			where $\mathbf{c}_{(a_1,a_2)}, \mathbf{c}_{(b_1,b_2)}\in \mathcal{C}_{D_4}$ with $(a_1,a_2)\ne (b_1,b_2)$.
			\item  [($\mathrm{\romannumeral 2}$)] Let $D_5=(D_{F,SQ}\times(D_{H,NSQ}\cup D_{H,0}))\cup ((D_{F,NSQ}\cup D_{F,0})\times D_{H,SQ})$, then $\mathcal{C}_{D_5}$ defined by $(1)$ is a $(p^{4m}, K_2)$ asymptotically optimal codebook with 
			{\small 	\[I_{\mathrm{max}}(\mathcal{C}_{D_5})=\frac{(p^m+1)^2}{2K_2},\] }
			where $K_2=\frac{p^{4m}-1}{2}$. Furthermore, the distribution of cross-correlation amplitudes of $\mathcal{C}_{D_5}$ is given by

			{\small \[|\mathbf{c}_{(a_1,a_2)}\mathbf{c}_{(b_1,b_2)}^H|=\begin{cases}
					\frac{p^m+ 1}{2K_2},\  (p^{2m}-1)p^{4m}\ \text{times},\\
					\frac{p^m- 1}{2K_2},\  (p^{2m}-1)p^{4m}\ \text{times},\\
					\frac{(p^m+ 1)^2}{2K_2}, \ \frac{(p^{2m}-1)^2}{4}p^{4m}\ \text{times},\\
					\frac{(p^m-1)^2}{2K_2}, \ \frac{(p^{2m}-1)^2}{4}p^{4m}\ \text{times},\\
					\frac{p^{2m}- 1}{2K_2}, \ \frac{(p^{2m}-1)^2}{2}p^{4m}\ \text{times},
				\end{cases}\]}
			where $\mathbf{c}_{(a_1,a_2)}, \mathbf{c}_{(b_1,b_2)}\in \mathcal{C}_{D_5}$ with $\mathbf{c}_{(a_1,a_2)}\ne \mathbf{c}_{(b_1,b_2)}$.
		\end{itemize}
	\end{theorem}
	\begin{proof}
		We only give the proof of ($\mathrm{\romannumeral 1}$), and the proof of ($\mathrm{\romannumeral 2}$) is similar.
		
		By Lemma \ref{Le6}, we get that $K_1=\frac{p^{4m}+1}{2}$. Let $\mathbf{c}_{(a_1,a_2)}$, $\mathbf{c}_{(b_1,b_2)}$ be two codewords in $\mathcal{C}_{D_4}$ with $(a_1,a_2)\ne (b_1,b_2)$, and let $\mu_a$ denote the additive character of $V_r^{(p)}$, defined by $\mu_a(x)=\zeta_p^{\langle a,x\rangle_r}$, $a,x\in V_r^{(p)}$. Then 	
		{\small 	\begin{align*}
				\mathbf{c}_{(a_1,a_2)}\mathbf{c}_{(b_1,b_2)}^H
				=\frac{1}{K_1}(\mu_{a_1-b_1}(D_{F,SQ})\mu_{a_2-b_2}(D_{H,SQ})+\mu_{a_1-b_1}(D_{F,NSQ}\cup D_{F,0})\mu_{a_2-b_2}(D_{H,NSQ}\cup D_{H,0})).
		\end{align*}}
		According to Lemmas \ref{Le6} and \ref{Le7}, we calculate the value of $\mathbf{c}_{(a_1,a_2)}\mathbf{c}_{(b_1,b_2)}^H$ by considering the following cases.
		\begin{itemize}
			\item [$\bullet$] If $a_1=b_1$ and $a_2\ne b_2$, then $H^*(a_2-b_2)\ne 0$, thus 
			{\small 	\begin{align*}
					\mathbf{c}_{(a_1,a_2)}\mathbf{c}_{(b_1,b_2)}^H&=\frac{1}{K_1}(-|D_{F,SQ}|\frac{\eta_m(H^*(a_2-b_2))p^m+1}{2}+(|D_{F,NSQ}|+|D_{F,0}|)\frac{\eta_m(H^*(a_2-b_2))p^m+1}{2})\\
					&=\frac{\eta_m(H^*(a_2-b_2))p^m+1}{2K_1}.
			\end{align*}}
			\item [$\bullet$] If $a_1\ne b_1$ and $a_2=b_2$, then $F^*(a_1-b_1)\ne 0$, thus
			{\small 	\begin{align*}
					\mathbf{c}_{(a_1,a_2)}\mathbf{c}_{(b_1,b_2)}^H
					&=\frac{1}{K_1}(-|D_{H,SQ}|\frac{\eta_m(F^*(a_1-b_1))p^m+1}{2}+(|D_{H,NSQ}|+|D_{H,0}|)\frac{\eta_m(F^*(a_1-b_1))p^m+1}{2})\\
					&=\frac{\eta_m(F^*(a_1-b_1))p^m+1}{2K_1}.
			\end{align*}}
			\item [$\bullet$] If $a_1\ne b_1$ and $a_2\ne b_2$, then $F^*(a_1-b_1)\ne 0$ and $H^*(a_2-b_2)\ne 0$, thus
			{\small 	\begin{align*}
					\mathbf{c}_{(a_1,a_2)}\mathbf{c}_{(b_1,b_2)}^H=\frac{(\eta_m(F^*(a_1-b_1))p^m+1)(\eta_m(H^*(a_2-b_2))p^m+1)}{2K_1}.
			\end{align*}}
		\end{itemize}
	Hence, we have that $I_{\mathrm{max}}(\mathcal{C}_{D_4})=\frac{(p^m+1)^2}{2K_1}$.	Obviously, $\mathbf{c}_{(a_1,a_2)}\ne \mathbf{c}_{(b_1,b_2)}$, so $|\mathcal{C}_{D_4}|=p^{4m}$,
		 then 
		{\small \begin{align*}
				\frac{I_{\mathrm{max}}(\mathcal{C}_{D_4})}{I_{\mathrm{W}}}=\sqrt{\frac{(p^m+1)^4}{p^{4m}+1}}.
		\end{align*}}
		Therefore, $\frac{I_{\mathrm{max}}(\mathcal{C}_{D_4})}{I_{\mathrm{W}}}\rightarrow 1$ if $p\rightarrow +\infty$ or $m\rightarrow +\infty$, which implies that the codebook $\mathcal{C}_{D_4}$ asymptotically achieves the Welch bound.
		
		By Lemma \ref{Le6} and the above discussion, we can easily obtain the distribution of cross-correlation amplitudes of $\mathcal{C}_{D_4}$.
	\end{proof}
	\begin{remark}
		\begin{itemize}
			\item  [($\mathrm{\romannumeral 1}$)]
			According to Theorems \ref{Th2} and \ref{Th3}, $\frac{I_{\mathrm{max}}(\mathcal{C}_{D_i})}{I_\mathrm{W}}\rightarrow 1$, $i=2,3,4,5$, if $p\rightarrow +\infty$ or $m\rightarrow +\infty$. If  $p$ is very small and $m\rightarrow +\infty$, then the codebook $\mathcal{C}_{D_i}$ also asymptotically achieves the Welch bound, while its alphabet size is very small.
			\item	[($\mathrm{\romannumeral 2}$)]From the proofs of Theorems \ref{Th2} and \ref{Th3}, the condition that $r=2m$ and all component functions of $F$ (respectively $H$) are weakly regular but not regular is necessary and sufficient for the codebook $\mathcal{C}_{D_i}$, $i=2,3,4,5$, to be  asymptotically optimal.
			\item [($\mathrm{\romannumeral 3}$)] We can also consider other sets, for example $(D_{F,\mathbb{F}_{p^m}^*}\times SQ)\cup (D_{F,0}\times (NSQ\cup\{0\}))
			$ and $(D_{F,\mathbb{F}_{p^m}^*	}\times D_{H,SQ})\cup (D_{F,0}\times (D_{H,NSQ}\cup D_{H,0}))$, but the obtained codebooks are not asymptotically optimal. Hence, it is an interesting problem to select other sets to construct asymptotically optimal codebooks.
		\end{itemize}
	\end{remark}
	\subsection{The second construction of codebooks}
	Let $S\subseteq \mathbb{F}_{p^m}$, $D\subseteq V_n^{(p)}$, $|D|=K$, and  
	{\small \[E_1=\{\mathbf{c}_{(a,b)}=\frac{1}{\sqrt{K}}(\lambda_a(F(x))\chi_b(x))_{x\in D}: a\in S, b\in V_n^{(p)}\},\]}
	where $F$ is a vectorial dual-bent function from $V_n^{(p)}$ to $\mathbb{F}_{p^m}$ that belongs to $\mathscr{F}.$
	In this subsection, we consider the $(N,K)$ codebook defined by  {\small \begin{equation}\mathcal{C}_{(S,D)}=\mathcal{E}_K\cup E_1. 
	\end{equation}}
	Choosing distinct sets $S$ and $D$, we give several families of asymptotically optimal codebooks in the following theorems.
	\begin{theorem}\label{Th4}
		Let $F:V_n^{(p)}\longrightarrow\mathbb{F}_{p^m}$ be a vectorial dual-bent function belonging to $\mathscr{F}$. Let $S_1=\mathbb{F}_{p^m}$ and $D_1=V_n^{(p)}$, then $\mathcal{C}_{(S_1,D_1)}$ defined by $(2)$ is a $(p^{n+m}+p^{n}, p^{n})$ asymptotically optimal codebook with 
		{\small \[I_{\mathrm{max}}(\mathcal{C}_{(S_1,D_1)})=p^{-\frac{n}{2}}.\] }
		Furthermore, for two different codewords $\mathbf{c}_1,\mathbf{c}_2\in\mathcal{C}_{(S_1,D_1)}$, the distribution of their cross-correlation amplitude is given as follows.
		{\small \[|\mathbf{c}_1\mathbf{c}_2^H|=\begin{cases}
				0,\ (p^{n+m}+p^n)(p^n-1)\ \text{times},\\
				p^{-\frac{n}{2}}, \ p^{2n+m}(p^m+1)\ \text{times}.
			\end{cases}\]}
		
	\end{theorem}
	\begin{proof} Obviously, $K=p^n$.
		For two codewords $\mathbf{c}_1,\mathbf{c}_2\in\mathcal{C}_{(S_1,D_1)}$, we calculate the value of $|\mathbf{c}_1\mathbf{c}_2^H|$ by considering the following cases.
		\begin{itemize}
			\item[$\bullet$] If $\mathbf{c}_1\ne \mathbf{c}_2\in\mathcal{E}_K$, then $|\mathbf{c}_1\mathbf{c}_2^H|=0$.
			\item[$\bullet$] If $\mathbf{c}_1, \mathbf{c}_2\in E_1$, assume that $\mathbf{c}_1=\mathbf{c}_{(a_1,b_1)}$ and $\mathbf{c}_2=\mathbf{c}_{(a_2,b_2)}$ with $(a_1,b_1)\ne (a_2,b_2)$, then
			{\small 	\begin{align*}
					\mathbf{c}_1\mathbf{c}_2^H=\frac{1}{p^n}\sum\limits_{x\in V_n^{(p)}}\lambda_{a_1-a_2}(F(x))\chi_{b_1-b_2}(x).
			\end{align*}}
			\begin{itemize}
				\item  [($\mathrm{\romannumeral 1}$)] If $a_1=a_2$ and $b_1\ne b_2$, then $|\mathbf{c}_1\mathbf{c}_2^H|=0$.
				\item  [($\mathrm{\romannumeral 2}$)] If $a_1\ne a_2$, then 
				$|\mathbf{c}_1\mathbf{c}_2^H|=|\varepsilon p^{-\frac{n}{2}} \lambda_{(a_1-a_2)^{1-d}}(F^*(b_1-b_2))|=p^{-\frac{n}{2}}.$
			\end{itemize}
			\item[$\bullet$] If $\mathbf{c}_1\in \mathcal{E}_{K}$ and $\mathbf{c}_2\in E_1$, then $|\mathbf{c}_1\mathbf{c}_2^H|=p^{-\frac{n}{2}}.$
		\end{itemize}
		Hence, we have that $I_{\mathrm{max}}\left(\mathcal{C}_{(S_1,D_1)}\right)=p^{-\frac{n}{2}}$. Obviously, $N=p^{n+m}+p^n$, then
		{\small \begin{align*}
				\frac{I_{\mathrm{max}}\left(\mathcal{C}_{(S_1,D_1)}\right)}{I_\mathrm{W}}=\sqrt{\frac{p^{n+m}+p^n-1}{p^{n+m}}}.
		\end{align*}}
		Therefore, $\frac{I_{\mathrm{max}}\left(\mathcal{C}_{(S_1,D_1)}\right)}{I_\mathrm{W}}\rightarrow 1$ if $p\rightarrow +\infty$ or $m\rightarrow +\infty$, which implies that the codebook $\mathcal{C}_{(S_1,D_1)}$ asymptotically achieves the Welch bound.
		By the above discussion, we can easily obtain the distribution of cross-correlation amplitudes of $\mathcal{C}_{(S_1,D_1)}$.
	\end{proof}
	\begin{remark}
		\begin{itemize}
			\item  [($\mathrm{\romannumeral 1}$)]
			According to Theorem \ref{Th4}, $\frac{I_{\mathrm{max}}\left(\mathcal{C}_{(S_1,D_1)}\right)}{I_\mathrm{W}}\rightarrow 1$ if $p\rightarrow +\infty$ or $m\rightarrow +\infty$. If $p$ is very small and $m\rightarrow +\infty$, then the codebook  $\mathcal{C}_{(S_1,D_1)}$ also asymptotically achieves the Welch bound, while its alphabet size is very small. 
			\item  [($\mathrm{\romannumeral 2}$)]
			Let $F(x,y)$ be a vectorial $p$-ary function from $\mathbb{F}_{p^m}\times\mathbb{F}_{p^m}$ to $\mathbb{F}_{p^m}$ defined by $F(x,y)=xy$, then $F(x,y)$ is a vectorial dual-bent function belonging to $\mathscr{F}$. Thus, the codebook constructed in \cite[Theorem 1]{Tian} can be obtained by Theorem \ref{Th4}.
		\end{itemize}
	\end{remark}
	\begin{theorem}\label{Th5}
		Let $F:V_n^{(p)}\longrightarrow\mathbb{F}_{p^m}$ be a vectorial dual-bent function belonging to $\mathscr{F}$, where $n=2m$, $m>1$. Let $S_2=\mathbb{F}_{p^m}$ and $D_2=V_n^{(p)}\setminus D_{F,0}$, then $\mathcal{C}_{(S_2,D_2)}$ defined by $(2)$ is a $(p^{3m}+K, K)$ asymptotically optimal codebook with 
		{\small 	\[I_{\mathrm{max}}(\mathcal{C}_{(S_2,D_2)})= \frac{p^m+1}{K},\] }
		where $K=(p^m-1)(p^m-\varepsilon)$, $\varepsilon=\pm1$. Furthermore, for two different codewords $\mathbf{c}_1,\mathbf{c}_2\in\mathcal{C}_{(S_2,D_2)}$, the distribution of their cross-correlation amplitude is given as follows.
		\begin{itemize}
			\item [$\bullet$] If $\varepsilon=1$, then
			{\small 	\[|\mathbf{c}_1\mathbf{c}_2^H|=\begin{cases}
					0,\ (p^m-1)^2(p^{2m}-2p^m)\ \text{times},\\
					\frac{1}{p^m-1},\ p^{3m}(p^m-1)(2p^m+1)\ \text{times},\\
					\frac{1}{(p^m-1)^2},\ 3p^{3m}(p^m-1)^2\ \text{times},\\
					\frac{p^m+1}{(p^m-1)^2},\ p^{3m}(p^m-1)^2(p^{m-1}-1)\ \text{times},\\
					\frac{|p^m\zeta_p^i+1|}{(p^m-1)^2}, \ p^{4m-1}(p^m-1)^2\ \text{times},
				\end{cases}\]}
			where $i=1,2,...,p-1$.
			\item [$\bullet$] If $\varepsilon=-1$, then
			{\small 	\[|\mathbf{c}_1\mathbf{c}_2^H|=\begin{cases}
					0,\ (p^{2m}-1)(p^{2m}-2)\ \text{times},\\
					\frac{1}{p^{2m}-1},\ p^{3m}(p^{2m}-1)\ \text{times},\\
					\frac{1}{p^m-1},\ p^{4m-1}(p^m-1)(p^m-p+1)\ \text{times},\\
					\frac{1}{\sqrt{p^{2m}-1}},\ 2p^{3m}(p^{2m}-1)\ \text{times},\\
					\frac{|p^m\zeta_p^i+1|}{p^{2m}-1},\ p^{4m-1}(p^{2m}-1)\ \text{times},
				\end{cases}\]}
			where $i=1,2,...,p-1$.
		\end{itemize}
		
	\end{theorem}
	\begin{proof}
		By Lemma \ref{Le6}, we easily get that $K=(p^m-1)(p^m-\varepsilon)$. According to Lemmas \ref{Le6} and \ref{Le7}, for two codewords $\mathbf{c}_1,\mathbf{c}_2\in\mathcal{C}_{(S_2,D_2)}$, we calculate the value of $|\mathbf{c}_1\mathbf{c}_2^H|$ by considering the following cases.
		\begin{itemize}
			\item[$\bullet$] If $\mathbf{c}_1\ne \mathbf{c}_2\in\mathcal{E}_K$, then $|\mathbf{c}_1\mathbf{c}_2^H|=0$.
			\item[$\bullet$] If $\mathbf{c}_1, \mathbf{c}_2\in E_1$, assume that $\mathbf{c}_1=\mathbf{c}_{(a_1,b_1)}$ and $\mathbf{c}_2=\mathbf{c}_{(a_2,b_2)}$ with $(a_1,b_1)\ne (a_2,b_2)$, then
			{\small 	\[\mathbf{c}_1\mathbf{c}_2^H=\frac{1}{K}\left(\sum\limits_{x\in V_n^{(p)}}\lambda_{a_1-a_2}(F(x))\chi_{b_1-b_2}(x)-\chi_{b_1-b_2}(D_{F,0})\right).\]}	
			\begin{itemize}
				\item  [($\mathrm{\romannumeral 1}$)] If $a_1=a_2$ and $b_1\ne b_2$, then 
				{\small $|\mathbf{c}_1\mathbf{c}_2^H|=\begin{cases}
						\frac{p^m-1}{K},\hspace{1cm} \text{if}\ F^*(b_1-b_2)=0,\\
						\frac{1}{K}, \hspace{1.5cm}\text{if}\ F^*(b_1-b_2)\ne 0.
					\end{cases}$}
				\item  [($\mathrm{\romannumeral 2}$)] If $a_1\ne a_2$ and $b_1=b_2$, then
				$|\mathbf{c}_1\mathbf{c}_2^H|=\frac{p^m-\varepsilon}{K}$.
				\item  [($\mathrm{\romannumeral 3}$)] If $a_1\ne a_2$ and $b_1\ne b_2$, then $
				\mathbf{c}_1\mathbf{c}_2^H=\frac{1}{K}(\varepsilon p^m \lambda_{(a_1-a_2)^{1-d}}(F^*(b_1-b_2))-\chi_{b_1-b_2}(D_{F,0})).
				$
				Hence, if $ F^*(b_1-b_2)=0$, then $|\mathbf{c}_1\mathbf{c}_2^H|=	\frac{1}{K}$; if $F^*(b_1-b_2)\ne 0$, then $|\mathbf{c}_1\mathbf{c}_2^H|=|\frac{p^m\lambda_{(a_1-a_2)^{1-d}}(F^*(b_1-b_2))+1}{K}|\le \frac{p^m+1}{K}.$ Note that for any $b_1,b_2\in V_n^{(p)}$ with $F^*(b_1-b_2)\ne 0$, there exist $a_1,a_2\in\mathbb{F}_{p^m}$ with $a_1\ne a_2$ such that $\lambda_{(a_1-a_2)^{1-d}}(F^*(b_1-b_2))=1$, then $|\frac{p^m\lambda_{(a_1-a_2)^{1-d}}(F^*(b_1-b_2))+1}{K}|= \frac{p^m+1}{K}$.
			\end{itemize}
			\item[$\bullet$] If $\mathbf{c}_1\in \mathcal{E}_{K}$ and $\mathbf{c}_2\in E_1$, then $|\mathbf{c}_1\mathbf{c}_2^H|=\frac{1}{\sqrt{K}}.$
		\end{itemize}
		Hence, we have that $I_{\mathrm{max}}(\mathcal{C}_{(S_2,D_2)})= \frac{p^m+1}{K}$. Obviously, $N=p^{3m}+(p^m-1)(p^m-\varepsilon)$, then
		{\small 	\begin{align*}
				\frac{I_{\mathrm{max}}\left(\mathcal{C}_{(S_2,D_2)}\right)}{I_\mathrm{W}}=\sqrt{\frac{(p^m+1)^2(p^{3m}+(p^m-1)(p^m-\varepsilon)-1)}{(p^m-1)(p^m-\varepsilon)p^{3m}}}.
		\end{align*}}
		Therefore, $\frac{I_{\mathrm{max}}\left(\mathcal{C}_{(S_2,D_2)}\right)}{I_\mathrm{W}}\rightarrow 1$ if $p\rightarrow +\infty$ or $m\rightarrow +\infty$, which implies that the codebook $\mathcal{C}_{(S_2,D_2)}$ asymptotically achieves the Welch bound. 	By Lemma \ref{Le6} and the above discussion, we can easily obtain the distribution of cross-correlation amplitudes of $\mathcal{C}_{(S_2,D_2)}$.
	\end{proof}
	\begin{theorem}\label{Th6}
		Let $F:V_n^{(p)}\longrightarrow\mathbb{F}_{p^m}$ be a vectorial dual-bent function belonging to $\mathscr{F}$, where $n=2m$, $m>1$,  all component functions of $F$ be regular, and $(F_{\alpha})^*=(F^*)_{\alpha}$, $\alpha\in\mathbb{F}_{p^m}^*$. Let $D_3=V_n^{(p)}\setminus (D_{F,i}\cup\{0\}) $, where $i\in\mathbb{F}_{p^m}^*$, and $S_3=\{a\in\mathbb{F}_{p^m}:\mathrm{Tr}_1^m(ai)=0\}$.
		Then $\mathcal{C}_{(S_3,D_3)}$ defined by $(2)$ is a $(p^{3m-1}+p^{2m}-p^m,p^{2m}-p^m)$ asymptotically optimal codebook with 
		{\small 	\[I_{\mathrm{max}}(\mathcal{C}_{(S_3,D_3)})= \frac{1}{p^m-1}.\]} Furthermore,  for two different codewords $\mathbf{c}_1,\mathbf{c}_2\in\mathcal{C}_{(S_3,D_3)}$, the distribution of their cross-correlation amplitude is given  as follows.
		
		{\small \[|\mathbf{c}_1\mathbf{c}_2^H|=\begin{cases}
				0,\ (p^{2m}-p^m)(p^{2m}-p^m-1)+p^{3m-1}(p^{2m}+p^{2m-1}-2p^m)\ \text{times},\\
				\frac{1}{p^m-1},\  p^{3m-1}(p^{3m-1}-p^{2m}-p^{2m-1}+2p^m-1)\ \text{times},\\
				\frac{1}{\sqrt{p^{2m}-p^m}}, \ 2p^{3m-1}(p^{2m}-p^m)\ \text{times}.
			\end{cases}\]}
	\end{theorem}
	\begin{proof}
		By Lemma \ref{Le6}, we easily get that $K=(p^m-1)^2+p^m-1=p^{2m}-p^m$. 
		
		According to Lemmas \ref{Le6} and \ref{Le8}, for two codewords $\mathbf{c}_1,\mathbf{c}_2\in\mathcal{C}_{(S_3,D_3)}$, we calculate the value of $|\mathbf{c}_1\mathbf{c}_2^H|$ by considering the following cases.
		\item[$\bullet$] If $\mathbf{c}_1\ne \mathbf{c}_2\in\mathcal{E}_K$, then $|\mathbf{c}_1\mathbf{c}_2^H|=0$.
		\item[$\bullet$] If $\mathbf{c}_1, \mathbf{c}_2\in E_1$, assume that $\mathbf{c}_1=\mathbf{c}_{(a_1,b_1)}$ and $\mathbf{c}_2=\mathbf{c}_{(a_2,b_2)}$ with $(a_1,b_1)\ne (a_2,b_2)$, then
		{\small  \begin{align*}\mathbf{c}_1\mathbf{c}_2^H
				&=\frac{1}{K}\left(\sum\limits_{x\in V_n^{(p)}}\lambda_{a_1-a_2}(F(x))\chi_{b_1-b_2}(x)-\chi_{b_1-b_2}(D_{F,i})-1\right).
		\end{align*}}
		\begin{itemize}
			\item  [($\mathrm{\romannumeral 1}$)] If $a_1=a_2$ and $b_1\ne b_2$, then {\small $|\mathbf{c}_1\mathbf{c}_2^H|=\begin{cases}
					\frac{1}{p^m-1}, \ &\text{if}\ F^*(b_1-b_2)=i,\\
					0, \ &\text{if}\ F^*(b_1-b_2)\ne i.
				\end{cases}$}
			\item  [($\mathrm{\romannumeral 2}$)] If $a_1\ne a_2$ and $b_1=b_2$, then
			$
			\mathbf{c}_1\mathbf{c}_2^H=\frac{1}{K}( p^m-(p^m-1)-1)=0.
			$ Thus, $|\mathbf{c}_1\mathbf{c}_2^H|=0$.
			\item  [($\mathrm{\romannumeral 3}$)] If $a_1\ne a_2$ and $b_1\ne b_2$, then $
			\mathbf{c}_1\mathbf{c}_2^H=\frac{1}{K}(p^m\lambda_{a_1-a_2}(F^*(b_1-b_2))-\chi_{b_1-b_2}(D_{F,i})-1).$
			Hence, if $F^*(b_1-b_2)=i$, then $|\mathbf{c}_1\mathbf{c}_2^H|=0$. If $F^*(b_1-b_2)\ne i$, then $|\mathbf{c}_1\mathbf{c}_2^H|=|\frac{1}{K}(p^m\lambda_{a_1-a_2}(F^*(b_1-b_2)))|=\frac{1}{p^m-1}.$
		\end{itemize}
		\item[$\bullet$] If $\mathbf{c}_1\in \mathcal{E}_{K}$ and $\mathbf{c}_2\in E_1$, then $|\mathbf{c}_1\mathbf{c}_2^H|=\frac{1}{\sqrt{p^{2m}-p^m}}.$
		
		Hence, we have that $I_{\mathrm{max}}(\mathcal{C}_{(S_3,D_3)})= \frac{1}{p^m-1}$. Note that $|S_3|=p^{m-1}$, so $N=p^{3m-1}+p^{2m}-p^m$, then
		{\small 	\begin{align*}
				\frac{I_{\mathrm{max}}\left(\mathcal{C}_{(S_3,D_3)}\right)}{I_\mathrm{W}}=\sqrt{\frac{p^{3m-1}+p^{2m}-p^m-1}{(p^{m}-1)p^{2m-1}}}.
		\end{align*}}
		Therefore, $\frac{I_{\mathrm{max}}\left(\mathcal{C}_{(S_3,D_3)}\right)}{I_\mathrm{W}}\rightarrow 1$ if $p\rightarrow +\infty$ or $m\rightarrow +\infty$, which implies that the codebook $\mathcal{C}_{(S_3,D_3)}$ asymptotically achieves the Welch bound.
		By Lemma \ref{Le6} and the above discussion, we can easily obtain the distribution of cross-correlation amplitudes of $\mathcal{C}_{(S_3,D_3)}$.	
	\end{proof}
	\begin{remark}
		\begin{itemize}
			\item  [($\mathrm{\romannumeral 1}$)]
			According to Theorems \ref{Th5} and \ref{Th6}, $\frac{I_{\mathrm{max}}\left(\mathcal{C}_{(S_i,D_i)}\right)}{I_\mathrm{W}}\rightarrow 1$, $i=2,3$, if $p\rightarrow +\infty$ or $m\rightarrow +\infty$. If $p$ is very small and $m\rightarrow +\infty$, then the codebook  $\mathcal{C}_{(S_i,D_i)}$ also asymptotically achieves the Welch bound, while its alphabet size is very small.
			\item  [($\mathrm{\romannumeral 2}$)]By the proof of Theorem \ref{Th5}, the condition that $n=2m$ is necessary and sufficient for the codebook $\mathcal{C}_{(S_2,D_2)}$ to be asymptotically optimal. Similarly, by the proof of Theorem \ref{Th6}, the condition that $n=2m$ and  all component functions of $F$ are regular is also necessary and sufficient for the codebook $\mathcal{C}_{(S_3,D_3)}$ to be asymptotically optimal.
			
		\end{itemize}
	\end{remark}
	\section{Constructions of codebooks based on additive and multiplicative characters}
	In this section, based on additive and multiplicative characters, we give three constructions of codebooks from vectorial dual-bent functions. By which, we obtain some asymptotically optimal codebooks with respect to the Welch bound.
	\subsection{The third construction of codebooks}
	Let $D\subseteq V_n^{(p)}$ with $|D|=K$, and $P$ be any permutation on $\{1,2,\dots,\frac{p^m-1}{2}\}$ with $P(\frac{p^m-1}{2})=\frac{p^m-1}{2}$. Let $F:V_n^{(p)}\longrightarrow \mathbb{F}_{p^m}$ be a vectorial dual-bent function belonging to $\mathscr{F}$, 
	{\small 	\begin{align*}
			E_1&=\{\mathbf{c}_{(0,a)}=\frac{1}{\sqrt{K}}(\chi_a(x))_{x\in D}:a\in V_n^{(p)}\},\\
			E_2&=\{\mathbf{c}_{(j,a)}=\frac{1}{|(\theta_j(F(x))\chi_a(x))_{x\in D}|}(\theta_j(F(x))\chi_a(x))_{x\in D}:1\le j\le \frac{p^m-1}{2}, a\in V_n^{(p)}\},\\
			E_3&=\{\widetilde{\mathbf{c}}_{(j,a)}=\frac{1}{|(\tau_j(F(x))\chi_a(x))_{x\in D}|}(\tau_j(F(x))\chi_a(x))_{x\in D}:1\le j\le \frac{p^m-3}{2}, a\in V_n^{(p)}\},
	\end{align*}}
	where {\small  \begin{align*}\theta_j(F(x))&=\begin{cases}
			\varphi_j(F(x)), \ \text{if}\ F(x)\in NSQ\ \text{or}\ F(x)=0,\\
			\varphi_{P(j)}(F(x)), \ \text{if}\ F(x)\in SQ,
		\end{cases} \\
 \tau_j(F(x))&=\begin{cases}
			-\varphi_j(F(x)), \ \text{if}\ F(x)\in NSQ\ \text{or}\ F(x)=0,\\
			\varphi_{P(j)}(F(x)), \ \text{if}\ F(x)\in SQ.
		\end{cases}\end{align*}}
	
	In this subsection, we consider the $(N,K)$ codebook defined by 
	{\small \begin{equation}
			\mathcal{C}^{(1)}_{D}=\bigcup_{i=1}^3E_i\cup\mathcal{E}_K.
	\end{equation}}
	Choosing different sets $D$, we obtain two families of asymptotically optimal codebooks in the following theorems.
	\begin{theorem}\label{Th7}
		Let $F:V_{n}^{(p)}\longrightarrow \mathbb{F}_{p^m}$ be a vectorial dual-bent function belonging to $\mathscr{F}$, where $(p,m)\ne (3,1)$,  all component functions of $F$ be regular when $p>3$ and $(F_{\alpha})^*=(F^*)_{\alpha}$, $\alpha\in\mathbb{F}_{p^m}^*$. Let $D_1=V_n^{(p)}$, then $\mathcal{C}^{(1)}_{{D_1}}$ defined by $(3)$ is a $(p^{n+m}, p^n)$ asymptotically optimal codebook with 
		{\small 	\[I_{\mathrm{max}}(\mathcal{C}^{(1)}_{{D_1}})=\frac{p^{\frac{n}{2}}}{p^n-T},\]}
		where $T=p^{n-m}-\varepsilon p^{\frac{n}{2}-m}+\varepsilon p^{\frac{n}{2}}$, $\varepsilon=\pm 1$ when $p=3$ and $\varepsilon=1$ when $p>3$.  Furthermore,  for two different codewords $\mathbf{c}_1,\mathbf{c}_2\in\mathcal{C}^{(1)}_{{D_1}}$, the distribution of their cross-correlation amplitude is given as follows.
		
		{\small \[|\mathbf{c}_1\mathbf{c}_2^H|=\begin{cases}
				0,\ 2p^n(p^n-1)+p^n(p^m-2)(p^m+1)T\ \text{times},
				\\
				\frac{p^{\frac{n}{2}}(p^m-1)}{p^m(p^n-T)}, \ p^n(p^m-2)(T-1)\ \text{times},\\
				\frac{p^{\frac{n}{2}}}{p^m(p^n-T)}, p^n(p^m-2)(p^n-T)\ \text{times},\\
				\frac{p^{\frac{n}{2}}}{p^n-T}, \ p^n(p^m-2)(p^m-3)(p^n-T)\ \text{times},\\
				p^{-\frac{n}{2}}, \ 2p^{2n}\ \text{times},\\
				\frac{1}{\sqrt{p^n-T}}, \ 4p^n(p^m-2)(p^n-T) \text{times}.
			\end{cases}\]}
	\end{theorem}
	\begin{proof}
		Obviously, $K=p^n$.	According to Lemma \ref{Le6}, we easily get that for $1\le j\le \frac{p^m-1}{2}$, $1\le k\le \frac{p^m-3}{2}$, $|(\theta_j(F(x))\chi_a(x))_{x\in D}|=|(\tau_k(F(x))\chi_a(x))_{x\in D}|=(p^m-1)(p^{n-m}-\varepsilon 
		p^{\frac{n}{2}-m})$. By Lemmas \ref{Le1}, \ref{Le6}, \ref{Le7} and \ref{Le10}, for two codewords $\mathbf{c}_1,\mathbf{c}_2\in\mathcal{C}^{(1)}_{{D_1}}$, we calculate the value of $|\mathbf{c}_1\mathbf{c}_2^H|$ by considering the following cases.
		\begin{itemize}
			\item [$\bullet$] If $\mathbf{c}_1\ne \mathbf{c}_2\in \mathcal{E}_K$, then $|\mathbf{c}_1\mathbf{c}_2^H|=0$.
			\item[$\bullet$] If $\mathbf{c}_1, \mathbf{c}_2 \in E_1$, assume that $\mathbf{c}_1=\mathbf{c}_{(0,a_1)}$ and $\mathbf{c}_2=\mathbf{c}_{(0,a_2)}$ with $a_1\ne a_2$, then $|\mathbf{c}_1\mathbf{c}_2^H|=0$.
			\item[$\bullet$] If $\mathbf{c}_1, \mathbf{c}_2 \in E_2$, assume that $\mathbf{c}_1=\mathbf{c}_{(j,a_1)}$ and $\mathbf{c}_2=\mathbf{c}_{(k,a_2)}$ with $(j,a_1)\ne (k,a_2)$, then
			{\small 	\begin{align*}
					\mathbf{c}_1\mathbf{c}_2^H&=\frac{1}{p^n-T}\left(\sum\limits_{x\in D_{F,NSQ}}\varphi_{j}\overline{\varphi_k}(F(x))\chi_{a_1-a_2}(x)+\sum\limits_{x\in D_{F,SQ}}\varphi_{P(j)}\overline{\varphi_{P(k)}}(F(x))\chi_{a_1-a_2}(x)\right).
			\end{align*}}
			\begin{itemize}
				\item [($\mathrm{\romannumeral 1}$)] If $j=k$ and $a_1\ne a_2$, then {\small $|\mathbf{c}_1\mathbf{c}_2^H|=\begin{cases}
						\frac{{p^{\frac{n}{2}}}(p^m-1)}{p^m(p^n-T)}, \ \text{if}\ F^*(a_1-a_2)=0,\\
						\frac{p^{\frac{n}{2}}}{p^m(p^n-T)}, \ \text{if}\ F^*(a_1-a_2)\ne0.
					\end{cases}$}
				\item [($\mathrm{\romannumeral 2}$)] If $j\ne k$ and $a_1=a_2$, then $|\mathbf{c}_1\mathbf{c}_2^H|=0$.
				\item [($\mathrm{\romannumeral 3}$)] If $j\ne k$ and $a_1\ne a_2$, then {\small $|\mathbf{c}_1\mathbf{c}_2^H|=\begin{cases}
						0,\hspace{1cm} \text{if}\ F^*(a_1-a_2)=0, \\
						\frac{p^{\frac{n}{2}}}{p^n-T},\hspace{0.4cm} \text{if}\ F^*(a_1-a_2)\ne 0.
					\end{cases}$}
			\end{itemize}
			\item[$\bullet$] If $\mathbf{c}_1, \mathbf{c}_2 \in E_3$, assume that $\mathbf{c}_1=\widetilde{\mathbf{c}}_{(j,a_1)}$ and $\mathbf{c}_2=\widetilde{\mathbf{c}}_{(k,a_2)}$ with $(j,a_1)\ne (k,a_2)$, then
			{\small \begin{align*}
					\mathbf{c}_1\mathbf{c}_2^H&=\frac{1}{p^n-T}\left(\sum\limits_{x\in D_{F,NSQ}}\varphi_{j}\overline{\varphi_k}(F(x))\chi_{a_1-a_2}(x)+\sum\limits_{x\in D_{F,SQ}}\varphi_{P(j)}\overline{\varphi_{P(k)}}(F(x))\chi_{a_1-a_2}(x)\right).
			\end{align*}}
			\begin{itemize}
				\item [($\mathrm{\romannumeral 1}$)] If $j=k$ and $a_1\ne a_2$, then {\small $|\mathbf{c}_1\mathbf{c}_2^H|=\begin{cases}
						\frac{{p^{\frac{n}{2}}}(p^m-1)}{p^m(p^n-T)}, \ \text{if}\ F^*(a_1-a_2)=0,\\
						\frac{p^{\frac{n}{2}}}{p^m(p^n-T)}, \ \text{if}\ F^*(a_1-a_2)\ne0.
					\end{cases}$}
				\item [($\mathrm{\romannumeral 2}$)] If $j\ne k$ and $a_1=a_2$, then $|\mathbf{c}_1\mathbf{c}_2^H|=0$.
				\item [($\mathrm{\romannumeral 3}$)] If $j\ne k$ and $a_1\ne a_2$, then {\small $|\mathbf{c}_1\mathbf{c}_2^H|=\begin{cases}
						0,\hspace{1cm} \text{if}\ F^*(a_1-a_2)=0, \\
						\frac{p^{\frac{n}{2}}}{p^n-T},\hspace{0.4cm} \text{if}\ F^*(a_1-a_2)\ne 0.
					\end{cases}$}
			\end{itemize}
			\item[$\bullet$] If $\mathbf{c}_1\in \mathcal{E}_K, \mathbf{c}_2 \in E_1$, then $|\mathbf{c}_1\mathbf{c}_2^H|=p^{-\frac{n}{2}}$.
			\item[$\bullet$] If $\mathbf{c}_1\in \mathcal{E}_K, \mathbf{c}_2 \in E_2$ (respectively $E_3$), then $|\mathbf{c}_1\mathbf{c}_2^H|=\frac{1}{\sqrt{p^n-T}}$ or $0$.
			\item[$\bullet$] If $\mathbf{c}_1\in E_1, \mathbf{c}_2 \in E_2$, assume that $\mathbf{c}_1=\mathbf{c}_{(0,a_1)}$ and $\mathbf{c}_2=\mathbf{c}_{(k,a_2)}$, then
			
			{\small \[|\mathbf{c}_1\mathbf{c}_2^H|=\frac{1}{\sqrt{p^n(p^n-T)}}\left( \sum\limits_{x\in D_{F,NSQ}}\overline{\varphi_{k}}(F(x))\chi_{a_1-a_2}(x)+\sum\limits_{x\in D_{F,SQ}}\overline{\varphi_{P(k)}}(F(x))\chi_{a_1-a_2}(x)\right).\]}
			\begin{itemize}
				\item  [($\mathrm{\romannumeral 1}$)] If $a_1=a_2$, then $|\mathbf{c}_1\mathbf{c}_2^H|=0$. 
				\item  [($\mathrm{\romannumeral 2}$)] If  $a_1\ne a_2$, then 
				{\small $|\mathbf{c}_1\mathbf{c}_2^H|=\begin{cases}
						0,\hspace{1.35cm} \text{if}\ F^*(a_1-a_2)=0,\\
						\frac{1}{\sqrt{p^n-T}},\hspace{0.5cm} \text{if}\ F^*(a_1-a_2)\ne 0.
					\end{cases}$}
			\end{itemize}
			\item[$\bullet$] If $\mathbf{c}_1\in E_1, \mathbf{c}_2 \in E_3$, assume that $\mathbf{c}_1=\mathbf{c}_{(0,a_1)}$ and $\mathbf{c}_2=\widetilde{\mathbf{c}}_{(k,a_2)}$, then {\small \[|\mathbf{c}_1\mathbf{c}_2^H|=\frac{1}{\sqrt{p^n(p^n-T)}}\left( -\sum\limits_{x\in D_{F,NSQ}}\overline{\varphi_{k}}(F(x))\chi_{a_1-a_2}(x)+\sum\limits_{x\in D_{F,SQ}}\overline{\varphi_{P(k)}}(F(x))\chi_{a_1-a_2}(x)\right).\]}
			\begin{itemize}
				\item  [($\mathrm{\romannumeral 1}$)] If $a_1=a_2$, then $|\mathbf{c}_1\mathbf{c}_2^H|=0$. 
				\item  [($\mathrm{\romannumeral 2}$)] If  $a_1\ne a_2$, then 
				{\small $|\mathbf{c}_1\mathbf{c}_2^H|=\begin{cases}
						0,\hspace{1.35cm} \text{if}\ F^*(a_1-a_2)=0,\\
						\frac{1}{\sqrt{p^n-T}}, \hspace{0.5cm} \text{if}\ F^*(a_1-a_2)\ne 0.
					\end{cases}$}
			\end{itemize}
			\item[$\bullet$] If $\mathbf{c}_1\in E_2, \mathbf{c}_2 \in E_3$, assume that $\mathbf{c}_1=\mathbf{c}_{(j,a_1)}$ and $\mathbf{c}_2=\widetilde{\mathbf{c}}_{(k,a_2)}$, then
			{\small \[|\mathbf{c}_1\mathbf{c}_2^H|=\frac{1}{p^n-T}\left( -\sum\limits_{x\in D_{F,NSQ}}\varphi_{j}\overline{\varphi_k}(F(x))\chi_{a_1-a_2}(x)+\sum\limits_{x\in D_{F,SQ}}\varphi_{P(j)}\overline{\varphi_{P(k)}}(F(x))\chi_{a_1-a_2}(x)\right).\]}
			\item  [($\mathrm{\romannumeral 1}$)] If $a_1=a_2$, then $|\mathbf{c}_1\mathbf{c}_2^H|=0$. 
			\item  [($\mathrm{\romannumeral 2}$)]	 If $a_1\ne a_2$, then {\small $|\mathbf{c}_1\mathbf{c}_2^H|=\begin{cases}
					0,\hspace{1cm} \text{if}\ F^*(a_1-a_2)=0, \\
					\frac{p^{\frac{n}{2}}}{p^n-T},\hspace{0.4cm} \text{if}\ F^*(a_1-a_2)\ne 0.
				\end{cases}$}
		\end{itemize}
		Hence, we have that $I_{\mathrm{max}}(\mathcal{C}^{(1)}_{D})=\frac{p^{\frac{n}{2}}}{p^n-T}$. Obviously, $N=p^{n+m}$, then  {\small \[\frac{I_{\mathrm{max}}(\mathcal{C}^{(1)}_{D})}{I_{\mathrm{W}}}=\sqrt{\frac{p^{n}(p^{n+m}-1)}{(p^m-1)^3(p^{n-m}-\varepsilon p^{\frac{n}{2}-m})^2}}.  \]}
		Since $m\le \frac{n}{2}$, then  $\frac{I_{\mathrm{max}}(\mathcal{C}^{(1)}_{D_1})}{I_{\mathrm{W}}}\rightarrow 1$ if $p\rightarrow +\infty$ or $m\rightarrow +\infty$, which implies that the codebook $\mathcal{C}^{(1)}_{D_1}$ asymptotically achieves the Welch bound. By Lemma \ref{Le6} and the above discussion, we can easily obtain the distribution of cross-correlation amplitudes of $\mathcal{C}^{(1)}_{D_1}$.
	\end{proof}
	\begin{theorem}\label{Th8}
		Let $F:V_{n}^{(p)}\longrightarrow \mathbb{F}_{p^m}$ be a vectorial dual-bent function belonging to $\mathscr{F}$ such that all component functions of $F$ are regular when $p>3$ and $(F_{\alpha})^*=(F^*)_{\alpha}$, $\alpha\in\mathbb{F}_{p^m}^*$. Let $D_2=V_n^{(p)}\setminus D_{F,0}$, then $\mathcal{C}^{(1)}_{{D_2}}$ defined by $(3)$ is a $(K+p^{n}(p^m-1),K)$ asymptotically optimal codebook with 
		{\small 	\[I_{\mathrm{max}}(\mathcal{C}^{(1)}_{{D_2}})=\frac{p^{\frac{n}{2}}}{K},\]}
		where $K=(p^m-1)(p^{n-m}-\varepsilon p^{\frac{n}{2}-m})$, $\varepsilon=\pm 1$ when $p=3$ and $\varepsilon=1$ when $p>3$.  Furthermore,  for two different codewords $\mathbf{c}_1,\mathbf{c}_2\in\mathcal{C}^{(1)}_{{D_2}}$, the distribution of their cross-correlation amplitude is given as follows.
		{\small 	\[|\mathbf{c}_1\mathbf{c}_2^H|=\begin{cases}
				0,\ K(K-1)+p^n(p^m-2)(p^m-1)(p^n-K)\ \text{times},
				\\
				\frac{p^{\frac{n}{2}}(p^m-1)}{p^mK}, \ p^n(p^m-1)(p^n-K-1)\ \text{times},\\
				\frac{p^{\frac{n}{2}}}{p^mK}, p^n(p^m-1)K\ \text{times},\\
				\frac{p^{\frac{n}{2}}}{K}, \ p^n(p^m-1)(p^m-2)K\ \text{times},\\
				\frac{1}{\sqrt{K}}, \ 2p^n(p^m-1)K\  \text{times}.
			\end{cases}\]}
	\end{theorem}
	\begin{proof}
		According to Lemma \ref{Le6}, we easily get $K=(p^m-1)(p^{n-m}-\varepsilon p^{\frac{n}{2}-m})$. By Lemmas \ref{Le6}, \ref{Le7} and \ref{Le10}, for two codewords $\mathbf{c}_1,\mathbf{c}_2\in\mathcal{C}^{(1)}_{{D_2}}$, we calculate the value of $|\mathbf{c}_1\mathbf{c}_2^H|$ by considering the following cases.
		\begin{itemize}
			\item [$\bullet$] If $\mathbf{c}_1\ne \mathbf{c}_2\in \mathcal{E}_K$, then $|\mathbf{c}_1\mathbf{c}_2^H|=0$.
			\item[$\bullet$] If $\mathbf{c}_1, \mathbf{c}_2 \in E_1$, assume that $\mathbf{c}_1=\mathbf{c}_{(0,a_1)}$ and $\mathbf{c}_2=\mathbf{c}_{(0,a_2)}$ with $a_1\ne a_2$, then
			{\small 	\[|\mathbf{c}_1\mathbf{c}_2^H|=\begin{cases}
					\frac{p^{\frac{n}{2}}(p^m-1)}{p^mK}, \ \text{if}\ F^*(a_1-a_2)=0,\\
					\frac{p^{\frac{n}{2}}}{p^mK}, \hspace{0.9cm} \text{if}\ F^*(a_1-a_2)\ne 0.
				\end{cases}\]}
			\item[$\bullet$] If $\mathbf{c}_1, \mathbf{c}_2 \in E_2$, assume that $\mathbf{c}_1=\mathbf{c}_{(j,a_1)}$ and $\mathbf{c}_2=\mathbf{c}_{(k,a_2)}$ with $(j,a_1)\ne (k,a_2)$, then
			{\small  \begin{align*}
					\mathbf{c}_1\mathbf{c}_2^H&=\frac{1}{K}\left(\sum\limits_{x\in D_{F,NSQ}}\varphi_{j}\overline{\varphi_k}(F(x))\chi_{a_1-a_2}(x)+\sum\limits_{x\in D_{F,SQ}}\varphi_{P(j)}\overline{\varphi_{P(k)}}(F(x))\chi_{a_1-a_2}(x)\right).
			\end{align*}}
			\begin{itemize}
				\item [($\mathrm{\romannumeral 1}$)] If $j=k$ and $a_1\ne a_2$, then {\small $|\mathbf{c}_1\mathbf{c}_2^H|=\begin{cases}
						\frac{{p^{\frac{n}{2}}}(p^m-1)}{p^mK}, \ \text{if}\ F^*(a_1-a_2)=0,\\
						\frac{p^{\frac{n}{2}}}{p^mK}, \hspace{0.9cm} \text{if}\ F^*(a_1-a_2)\ne0.
					\end{cases}$}
				\item [($\mathrm{\romannumeral 2}$)] If $j\ne k$ and $a_1=a_2$, then $|\mathbf{c}_1\mathbf{c}_2^H|=0$.
				\item [($\mathrm{\romannumeral 3}$)] If $j\ne k$ and $a_1\ne a_2$, then {\small $|\mathbf{c}_1\mathbf{c}_2^H|=\begin{cases}
						0,\hspace{1cm} \text{if}\ F^*(a_1-a_2)=0, \\
						\frac{p^{\frac{n}{2}}}{K},\hspace{0.7cm} \text{if}\ F^*(a_1-a_2)\ne 0.
					\end{cases}$}
			\end{itemize}
			\item[$\bullet$] If $\mathbf{c}_1, \mathbf{c}_2 \in E_3$, assume that $\mathbf{c}_1=\widetilde{\mathbf{c}}_{(j,a_1)}$ and $\mathbf{c}_2=\widetilde{\mathbf{c}}_{(k,a_2)}$ with $(j,a_1)\ne (k,a_2)$, then
			{\small  \begin{align*}
					\mathbf{c}_1\mathbf{c}_2^H&=\frac{1}{K}\left(\sum\limits_{x\in D_{F,NSQ}}\varphi_{j}\overline{\varphi_k}(F(x))\chi_{a_1-a_2}(x)+\sum\limits_{x\in D_{F,SQ}}\varphi_{P(j)}\overline{\varphi_{P(k)}}(F(x))\chi_{a_1-a_2}(x)\right).
			\end{align*}}
			\begin{itemize}
				\item [($\mathrm{\romannumeral 1}$)] If $j=k$ and $a_1\ne a_2$, then {\small $|\mathbf{c}_1\mathbf{c}_2^H|=\begin{cases}
						\frac{{p^{\frac{n}{2}}}(p^m-1)}{p^mK}, \ \text{if}\ F^*(a_1-a_2)=0,\\
						\frac{p^{\frac{n}{2}}}{p^mK}, \hspace{0.9cm} \text{if}\ F^*(a-b)\ne0.
					\end{cases}$}
				\item [($\mathrm{\romannumeral 2}$)] If $j\ne k$ and $a_1=a_2$, then $|\mathbf{c}_1\mathbf{c}_2^H|=0$.
				\item [($\mathrm{\romannumeral 3}$)] If $j\ne k$ and $a_1\ne a_2$, then {\small $|\mathbf{c}_1\mathbf{c}_2^H|=\begin{cases}
						0,\hspace{1cm} \text{if}\ F^*(a_1-a_2)=0, \\
						\frac{p^{\frac{n}{2}}}{K},\hspace{0.7cm} \text{if}\ F^*(a_1-a_2)\ne 0.
					\end{cases}$}
			\end{itemize}
			\item[$\bullet$] If $\mathbf{c}_1\in \mathcal{E}_K, \mathbf{c}_2 \in E_1$, then $|\mathbf{c}_1\mathbf{c}_2^H|=\frac{1}{\sqrt{K}}$.
			\item[$\bullet$] If $\mathbf{c}_1\in \mathcal{E}_K, \mathbf{c}_2 \in E_2$ (respectively $E_3$), then $|\mathbf{c}_1\mathbf{c}_2^H|=\frac{1}{\sqrt{K}}$.
			\item[$\bullet$] If $\mathbf{c}_1\in E_1, \mathbf{c}_2 \in E_2$, assume that $\mathbf{c}_1=\mathbf{c}_{(0,a_1)}$ and $\mathbf{c}_2=\mathbf{c}_{(k,a_2)}$, then {\small \[|\mathbf{c}_1\mathbf{c}_2^H|=\frac{1}{K}\left( \sum\limits_{x\in D_{F,NSQ}}\overline{\varphi_{k}}(F(x))\chi_{a_1-a_2}(x)+\sum\limits_{x\in D_{F,SQ}}\overline{\varphi_{P(k)}}(F(x))\chi_{a_1-a_2}(x)\right).\]}
			\begin{itemize}
				\item  [($\mathrm{\romannumeral 1}$)] If $a_1=a_2$, then $|\mathbf{c}_1\mathbf{c}_2^H|=0$. 
				\item  [($\mathrm{\romannumeral 2}$)] If  $a_1\ne a_2$, then 
				{\small $|\mathbf{c}_1\mathbf{c}_2^H|=\begin{cases}
						0,\hspace{0.4cm} \text{if}\ F^*(a_1-a_2)=0,\\
						\frac{p^{\frac{n}{2}}}{K}, \ \text{if}\ F^*(a_1-a_2)\ne 0.
					\end{cases}$}
			\end{itemize}
			\item[$\bullet$] If $\mathbf{c}_1\in E_1, \mathbf{c}_2 \in E_3$, assume that $\mathbf{c}_1=\mathbf{c}_{(0,a_1)}$ and $\mathbf{c}_2=\widetilde{\mathbf{c}}_{(k,a_2)}$, then {\small \[|\mathbf{c}_1\mathbf{c}_2^H|=\frac{1}{K}\left( -\sum\limits_{x\in D_{F,NSQ}}\overline{\varphi_{k}}(F(x))\chi_{a_1-a_2}(x)+\sum\limits_{x\in D_{F,SQ}}\overline{\varphi_{P(k)}}(F(x))\chi_{a_1-a_2}(x)\right).\]}
			\begin{itemize}
				\item  [($\mathrm{\romannumeral 1}$)] If $a_1=a_2$, then $|\mathbf{c}_1\mathbf{c}_2^H|=0$. 
				\item  [($\mathrm{\romannumeral 2}$)] If  $a_1\ne a_2$, then 
				{\small $|\mathbf{c}_1\mathbf{c}_2^H|=\begin{cases}
						0,\hspace{0.4cm} \text{if}\ F^*(a_1-a_2)=0,\\
						\frac{p^{\frac{n}{2}}}{K}, \ \text{if}\ F^*(a_1-a_2)\ne 0.
					\end{cases}$}
			\end{itemize}
			\item[$\bullet$] If $\mathbf{c}_1\in E_2, \mathbf{c}_2 \in E_3$, assume that $\mathbf{c}_1=\mathbf{c}_{(j,a_1)}$ and $\mathbf{c}_2=\widetilde{\mathbf{c}}_{(k,a_2)}$, then
			{\small  \[|\mathbf{c}_1\mathbf{c}_2^H|=\frac{1}{K}\left( -\sum\limits_{x\in D_{F,NSQ}}\varphi_{j}\overline{\varphi_k}(F(x))\chi_{a_1-a_2}(x)+\sum\limits_{x\in D_{F,SQ}}\varphi_{P(j)}\overline{\varphi_{P(k)}}(F(x))\chi_{a_1-a_2}(x)\right).\]}
			\item  [($\mathrm{\romannumeral 1}$)] If $a_1=a_2$, then $|\mathbf{c}_1\mathbf{c}_2^H|=0$. 
			\item  [($\mathrm{\romannumeral 2}$)]	 If $a_1\ne a_2$, then {\small $|\mathbf{c}_1\mathbf{c}_2^H|=\begin{cases}
					0,\hspace{0.7cm} \text{if}\ F^*(a_1-a_2)=0, \\
					\frac{p^{\frac{n}{2}}}{K},\hspace{0.4cm} \text{if}\ F^*(a_1-a_2)\ne 0.
				\end{cases}$}
		\end{itemize}
		Hence, we have that $I_{\mathrm{max}}(\mathcal{C}^{(1)}_{D_2})=\frac{p^{\frac{n}{2}}}{K}$. Obviously, $N=(p^m-1)(p^n+p^{n-m}-\varepsilon p^{\frac{n}{2}-m})$, then  {\small \[\frac{I_{\mathrm{max}}(\mathcal{C}^{(1)}_{D_2})}{I_{\mathrm{W}}}=\sqrt{\frac{(p^m-1)(p^n+p^{n-m}-\varepsilon p^{\frac{n}{2}-m})-1} {(p^m-1)^2(p^{n-m}-\varepsilon p^{\frac{n}{2}-m})}}.  \]}
		Since $m\le \frac{n}{2}$, then  $\frac{I_{\mathrm{max}}(\mathcal{C}^{(1)}_{D_2})}{I_{\mathrm{W}}}\rightarrow 1$ if $p\rightarrow +\infty$ or $m\rightarrow +\infty$, which implies that the codebook $\mathcal{C}^{(1)}_{D_2}$ asymptotically achieves the Welch bound. By Lemma \ref{Le6} and the above discussion, we can easily obtain the distribution of cross-correlation amplitudes of $\mathcal{C}^{(1)}_{D_2}$.
	\end{proof}
	\begin{remark}\label{Re6}
		\begin{itemize}
			\item  [($\mathrm{\romannumeral 1}$)] In Theorem \ref{Th7},  we choose two different  permutations $P$ and $\widetilde{P}$, and denote the corresponding codebooks by 
			$\mathcal{C}^{(1)}_{D_1}$  and $\widetilde{\mathcal{C}}^{(1)}_{D_1}$. Since there exists $1\le j_0\le \frac{p^m-3}{2}$ such that $P(j_0)\ne \widetilde{P}(j_0)$, then we have that the codeword $\mathbf{c}_{(j_0,a)}\in \mathcal{C}^{(1)}_{D_1}$ is not contained in $\widetilde{\mathcal{C}}^{(1)}_{D_2}$. Hence, the codebooks $\mathcal{C}^{(1)}_{D_1}$  and $\widetilde{\mathcal{C}}^{(1)}_{D_1}$ are different. Similarly, in Theorem \ref{Le8}, if we choose two different permutations, then the obtained codebooks are also different.
			\item  [($\mathrm{\romannumeral 2}$)] In particular, if $P$ is the identity mapping, then the condition that all component functions of $F$ are regular  when $p>3$ and $(F_{\alpha})^*=(F^*)_{\alpha}$, $\alpha\in\mathbb{F}_{p^m}^*$, is not necessary. According to Lemma \ref{Le9} and the proofs of Theorems \ref{Th7} and \ref{Th8}, we can get a $(p^{n+m},p^n)$ asymptotically optimal codebook $\mathcal{C}$ with $I_{\mathrm{max}}(\mathcal{C})=\frac{p^{\frac{n}{2}}}{(p^m-1)(p^{n-m}-\varepsilon p^{\frac{n}{2}-m})}$, and a $((p^m-1)(p^n+p^{n-m}-\varepsilon p^{\frac{n}{2}-m}), (p^m-1)(p^{n-m}-\varepsilon p^{\frac{n}{2}-m}))$ asymptotically optimal codebook $\mathcal{C}$ with $I_{\mathrm{max}}(\mathcal{C})=\frac{p^{\frac{n}{2}}}{(p^m-1)(p^{n-m}-\varepsilon p^{\frac{n}{2}-m})}$, where $\varepsilon=\pm 1$.
		\end{itemize}
		
	\end{remark}
	\subsection{The fourth construction of codebooks}
	Let $F:V_n^{(p)}\longrightarrow \mathbb{F}_{p^m}$ be a vectorial dual-bent function belonging to $\mathscr{F}$, $D=\{(x,y,z)\in\mathbb{F}_{p^n}\times \mathbb{F}_{p^m}\times\mathbb{F}_{p^m}:z=yF(x)\}$ and 
	{\small \begin{align*}
			E_1=\{\mathbf{c}_{(j,a,b)}=\frac{1}{|(\varphi_j(z)\lambda_a(y)\chi_b(x))_{(x,y,z)\in D }|}(\varphi_j(z)\lambda_a(y)\chi_b(x))_{(x,y,z)\in D }:0\le j \le p^m-2, 
			a\in \mathbb{F}_{p^m}, b\in V_n^{(p)}\}.
	\end{align*}} In this subsection, we consider the $(N,K)$ codebook defined by
	{\small \begin{equation}\mathcal{C}^{(2)}=\mathcal{E}_K\cup E_1,
	\end{equation}}
	where $K=p^{n+m}$. 
	\begin{theorem}\label{Th9}
		Let $F:V_n^{(p)}\longrightarrow \mathbb{F}_{p^m}$ be a vectorial dual-bent function belonging to $\mathscr{F}$, where $(p,m)\ne (3,1)$, then $\mathcal{C}^{(2)}$ defined by $(4)$ is a $(p^{n+2m}, p^{n+m})$ asymptotically optimal codebook with 
		{\small \[I_{\mathrm{max}}(\mathcal{C}^{(2)})=
			\frac{p^{\frac{n+m}{2}}}{(p^m-1)(p^n-T)},
			\]}
		where $T=p^{n-m}-\varepsilon p^{\frac{n}{2}-m}+\varepsilon p^{\frac{n}{2}}$, $\varepsilon=\pm 1$. Furthermore, for two different codewords $\mathbf{c}_1, \mathbf{c}_2\in \mathcal{C}^{(2)}$, the distribution of their cross-correlation amplitude is given as follows.
		{\small 	\[|\mathbf{c}_1\mathbf{c}_2^H|=\begin{cases}
				0, \ p^{n+m}(2p^{n+2m}-p^{n+m}-2p^n-p^m+T(p^{2m}-1)(p^m-2))\ \text{times},\\
				\frac{p^{\frac{n+m}{2}}}{(p^m-1)(p^n-T)}, \ p^{n+m}(p^m-1)(p^m-2)(p^m-3)(p^n-T)\  \text{times},\\
				\frac{1}{\sqrt{(p^m-1)(p^n-T)}}, \ 4p^{n+m}(p^m-1)(p^m-2)(p^n-T)\ \text{times},\\
				p^{-\frac{n+m}{2}}, \  2p^{2n+2m}\ \text{times}.
			\end{cases}\]}
	\end{theorem}
	\begin{proof}
	 Denote $|(\varphi_j(z)\lambda_a(y)\chi_b(x))_{(x,y,z)\in D }|$ by $T_j$. If $j=0$, then $T_0=\sqrt{p^{n+m}}$;  if $j\ne 0$, then according to Lemma \ref{Le6}, $T_j=(p^m-1)\sqrt{p^{n-m}-\varepsilon p^{\frac{n}{2}-m}}.$  By Lemmas \ref{Le1}, \ref{Le6}, and \ref{Le9}, for two codewords $\mathbf{c}_1, \mathbf{c}_2\in \mathcal{C}^{(2)}$, we calculate the value of $|\mathbf{c}_1\mathbf{c}_2^H|$ by considering the following cases.
		\begin{itemize}
			\item [$\bullet$] If $\mathbf{c}_1\ne \mathbf{c}_2\in\mathcal{E}_K$, then $|\mathbf{c}_1\mathbf{c}_2^H|=0$.
			\item[$\bullet$] If $\mathbf{c}_1,\mathbf{c}_2\in E_1$, assume that $\mathbf{c}_1=\mathbf{c}_{(i,a_1,b_1)}$ and $\mathbf{c}_2=\mathbf{c}_{(j,a_2,b_2)}$ with $(i,a_1,b_1)\ne (j,a_2,b_2)$, then 
			{\small 	\[\mathbf{c}_1\mathbf{c}_2^H=\frac{1}{T_iT_j}\sum\limits_{y\in\mathbb{F}_{p^m}}\varphi_{i}\overline{\varphi_j}(y)\lambda_{a_1-a_2}(y)\sum\limits_{x\in V_n^{(p)}}\varphi_{i}\overline{\varphi_j}(F(x))\chi_{b_1-b_2}(x).\]}
			\begin{itemize}
				\item [($\mathrm{\romannumeral 1}$)] If $i=j$, then $a_1\ne a_2$ or $b_1\ne b_2$, so $|\mathbf{c}_1\mathbf{c}_2^H|=0$.
				\item [($\mathrm{\romannumeral 2}$)] If $i\ne j$, $a_1= a_2$ or $b_1=b_2$, note that $\sum\limits_{x\in V_n^{(p)}}\varphi_{i}\overline{\varphi_j}(F(x))=(p^{n-m}-\varepsilon p^{\frac{n}{2}-m})\sum\limits_{c\in\mathbb{F}_{p^m}^*}\varphi_{i}\overline{\varphi_j}(c)=0$, then $|\mathbf{c}_1\mathbf{c}_2^H|=0$.
				\item [($\mathrm{\romannumeral 3}$)] If $i\ne j$, $a_1\ne a_2$ and $b_1\ne b_2$, then {\small $|\mathbf{c}_1\mathbf{c}_2^H|=\begin{cases}
						0,\hspace{0.8cm} \text{if}\ F^*(b_1-b_2)=0,\\
						\frac{p^{\frac{n+m}{2}}}{T_iT_j}, \ \text{if}\ F^*(b_1-b_2)\ne 0.
					\end{cases}$}
			\end{itemize}
			\item[$\bullet$] If $\mathbf{c}_1\in\mathcal{E}_K,\mathbf{c}_2\in E_1$, then {\small $|\mathbf{c}_1\mathbf{c}_2^H|=0,\ \text{or}\ p^{-\frac{n+m}{2}}, \ \text{or}\ \frac{1}{\sqrt{(p^m-1)^2(p^{n-m}-\varepsilon p^{\frac{n}{2}-m})}}$.}
		\end{itemize}
		
		Hence, we have that {\small \[I_{\mathrm{max}}(\mathcal{C}^{(2)})=
			\frac{p^{\frac{n+m}{2}}}{(p^m-1)(p^{n}-T)}.
			\] } Obviously, $N=p^{n+2m}$, then
		{\small \[\frac{I_{\mathrm{max}}(\mathcal{C}^{(2)})}{I_{\mathrm{W}}}= \sqrt{\frac{p^{n+m}(p^{n+2m}-1)}{(p^m-1)^5(p^{n-m}-\varepsilon p^{\frac{n}{2}-m})^2}}.\]}
		Since $m\le \frac{n}{2}$, then $\frac{I_{\mathrm{max}}(\mathcal{C}^{(2)})}{I_{\mathrm{W}}}\rightarrow 1$ if $p\rightarrow +\infty$ or $m\rightarrow +\infty$, which implies that the codebook $\mathcal{C}^{(2)}$ asymptotically achieves the Welch bound. By Lemma \ref{Le6} and the above discussion, we can easily obtain the distribution of cross-correlation amplitudes of $\mathcal{C}^{(2)}$.
	\end{proof}
	\subsection{The fifth construction of codebooks}
	Let $F:V_n^{(p)}\longrightarrow \mathbb{F}_{p^m}$ be a vectorial dual-bent function belonging to $\mathscr{F}$, $D=\{(x,y,z)\in\mathbb{F}_{p^n}\times \mathbb{F}_{p^m}\times\mathbb{F}_{p^m}:z=(1-y)F(x)\}$ and 
	{\small 	\[E_1=\{\mathbf{c}_{(i,j,a)}=\frac{1}{|(\varphi_i(y)\varphi_j(z)\chi_a(x))_{(x,y,z)\in D }|}(\varphi_i(y)\varphi_j(z)\chi_a(x))_{(x,y,z)\in D }:0\le i,j \le p^m-2, a\in V_n^{(p)}\}.\]} In this subsection, we consider the $(N,K)$ codebook defined by
	{\small \begin{equation}
			\mathcal{C}^{(3)}=\mathcal{E}_K\cup E_1,
	\end{equation}}
	where $K=p^{n+m}$. 
	\begin{theorem}\label{Th10}
		Let $F:V_n^{(p)}\longrightarrow \mathbb{F}_{p^m}$ be a vectorial dual-bent function belonging to $\mathscr{F}$, where $(p,m)\ne (3,1)$, then $\mathcal{C}^{(3)}$ defined by $(5)$ is a $(p^{n}(p^{2m}-p^m+1), p^{n+m})$ asymptotically optimal codebook with 
		{\small 	\[I_{\mathrm{max}}(\mathcal{C}^{(3)})=\frac{p^\frac{n+m}{2}}{(p^m-2)(p^n-T)},
			\]}
		where $T=p^{n-m}-\varepsilon p^{\frac{n}{2}-m}+\varepsilon p^{\frac{n}{2}}$, $\varepsilon=\pm 1$. Furthermore, for two different codewords $\mathbf{c}_1, \mathbf{c}_2\in \mathcal{C}^{(3)}$, the distribution of their cross-correlation amplitude is given as follows.
		{\small \[|\mathbf{c}_1\mathbf{c}_2^H|=\begin{cases}
				0,\ p^n(2p^{n+m}(p^{2m}-p^m-2)+5p^n-p^{m}(p^m-1)+T(p^{2m}(p^m-3)(p^m-1)+6p^m-8)-1)\ \text{times},\\
				\frac{p^{\frac{n+m}{2}}}{(p^m-2)(p^n-T)}, \ p^n(p^m-2)(p^m-3)^3(p^n-T)\ \text{times},\\
				\frac{1}{\sqrt{(p^m-2)(p^n-T)}}, \ 2p^n(p^{2m}-3p^m+1)(p^m-2)(p^n-T)\ \text{times},\\
				\frac{p^{\frac{m}{2}}}                                     {(p^m-1)\sqrt{p^n-T}},\ 2p^n(p^m-2)(p^m-3)(p^n-T)\ \text{times},\\
				\frac{p^{\frac{n+m}{2}}}{(p^n-T)\sqrt{(p^m-1)(p^m-2)}}, \ 2p^n(p^m-2)(p^m-3)^2(p^n-T)\ \text{times},\\
				\frac{p^{\frac{m}{2}}}{\sqrt{(p^m-1)(p^m-2)(p^n-T)}}, \ 2p^n(p^m-2)(p^m-3)^2(p^n-T)\ \text{times},\\
				\frac{p^{\frac{n}{2}}}{(p^m-2)(p^n-T)}, \ p^n(p^m-2)(p^m-3)^2(p^n-T)\ \text{times},\\
				\frac{1}{\sqrt{p^m(p^m-2)(p^n-T)}},\ 2p^n(p^m-2)(p^n-T)\ \text{times},\\
				\frac{1}{(p^m-1)\sqrt{p^n-T}},\ 2p^n(p^m-2)(p^n-T)\ \text{times},\\
				\frac{p^{\frac{n}{2}}}{(p^n-T)\sqrt{(p^m-1)(p^m-2)}},\ 2p^n(p^m-2)(p^m-3)(p^n-T)\ \text{times},\\
				\frac{1}{\sqrt{(p^m-1)(p^m-2)(p^n-T)}}, \ 2p^n(p^m-2)(p^m-3)(p^n-T)\ \text{times},\\
				p^{-\frac{n+m}{2}},\ 2p^{2n+m}\ \text{times},\\
				\frac{1}{\sqrt{(p^m-1)(p^n-T)}},\ 2p^n(p^m-1)(p^m-2)(p^n-T)\ \text{times},\\
				\frac{1}{\sqrt{p^n(p^m-1)}},\ 2p^{2n}(p^m-1)(p^m-2)\ \text{times}.
			\end{cases}\]}
	\end{theorem}
	\begin{proof}
		 Denote $|(\varphi_i(y)\varphi_j(z)\chi_a(x))_{(x,y,z)\in D }|$ by $T_{(i,j)}$. If $i=j=0$, then $T_{(0,0)}=\sqrt{p^{n+m}}$; if $i=0$, $j\ne0$, then $T_{(0,j)}=(p^m-1)\sqrt{p^{n-m}-\varepsilon p^{\frac{n}{2}-m}}$; if $i\ne 0$, $j=0$, then $T_{(i,0)}=\sqrt{p^n(p^m-1)}$; if $i\ne0$, $j\ne 0$, then $T_{(i,j)}=\sqrt{(p^m-1)(p^m-2)(p^{n-m}-\varepsilon p^{\frac{n}{2}-m})}$. By Lemmas \ref{Le1}, \ref{Le4} and \ref{Le9}, for two codewords $\mathbf{c}_1, \mathbf{c}_2\in \mathcal{C}^{(3)}$, we calculate the value of $|\mathbf{c}_1\mathbf{c}_2^H|$ by considering the following cases.
		\begin{itemize}
			\item [$\bullet$] If $\mathbf{c}_1\ne\mathbf{c}_2\in\mathcal{E}_K$, then $|\mathbf{c}_1\mathbf{c}_2^H|=0$.
			\item[$\bullet$] If $\mathbf{c}_1,\mathbf{c}_2\in E_1$, assume that $\mathbf{c}_1=\mathbf{c}_{(i,j,a_1)}$ and $\mathbf{c}_2=\mathbf{c}_{(r,s,a_2)}$, then
			{\small 	\[|\mathbf{c}_1\mathbf{c}_2^H|=\frac{1}{T_{(i,j)}T_{(r,s)}}\sum\limits_{y\in\mathbb{F}_{p^m}}\varphi_{i}\overline{\varphi_r}(y)\varphi_{j}\overline{\varphi_s}(1-y)\sum\limits_{x\in V_n^{(p)}}\varphi_{j}\overline{\varphi_s}(F(x))\chi_{a_1-a_2}(x).\]}
			\begin{itemize}
				\item [($\mathrm{\romannumeral 1}$)] If $a_1=a_2$, one of $\varphi_{i}\overline{\varphi_r}$ and $\varphi_{j}\overline{\varphi_s}$ is trivial and the other is nontrivial, then $|\mathbf{c}_1\mathbf{c}_2^H|=0$.
				\item [($\mathrm{\romannumeral 2}$)] If $a_1=a_2$, $i\ne r$ and $j\ne s$, then $\sum\limits_{x\in V_n^{(p)}}\varphi_{j}\overline{\varphi_s}(F(x))=0$, so $|\mathbf{c}_1\mathbf{c}_2^H|=0$.
				\item [($\mathrm{\romannumeral 3}$)] If $a_1\ne a_2$, $i=r$ and $j=s$, then  $|\mathbf{c}_1\mathbf{c}_2^H|=0$.
				\item [($\mathrm{\romannumeral 4}$)] If $a_1\ne a_2$, one of $\varphi_{i}\overline{\varphi_r}$ and $\varphi_{j}\overline{\varphi_s}$ is trivial and the other is nontrivial, then $|\mathbf{c}_1\mathbf{c}_2^H|=0$.
				\item [($\mathrm{\romannumeral 5}$)] If $a_1\ne a_2$, $i\ne r$, $j\ne s$ and $\varphi_{i}\overline{\varphi_r}\varphi_{j}\overline{\varphi_s}$ is nontrivial, then 
				
				{\small \[|\mathbf{c}_1\mathbf{c}_2^H|=\begin{cases}
						0,\hspace{1.4cm} \text{if}\ F^*(a_1-a_2)=0,\\
						\frac{p^{\frac{n+m}{2}}}{T_{(i,j)}T_{(r,s)}},\ \text{if}\ F^*(a_1-a_2)\ne 0.
					\end{cases}\]}
				\item [($\mathrm{\romannumeral 6}$)] If $a_1\ne a_2$, $i\ne r$, $j\ne s$ and $\varphi_{i}\overline{\varphi_r}\varphi_{j}\overline{\varphi_s}$ is trivial, then {\small \[|\mathbf{c}_1\mathbf{c}_2^H|=\begin{cases}
						0,\hspace{1.4cm} \text{if}\ F^*(a_1-a_2)=0,\\
						\frac{p^{\frac{n}{2}}}{T_{(i,j)}T_{(r,s)}},\ \text{if}\ F^*(a_1-a_2)\ne 0.
					\end{cases}\]}
			\end{itemize}
			\item[$\bullet$] If $\mathbf{c}_1\in \mathcal{E}_K$, $\mathbf{c}_2\in E_1$, then $|\mathbf{c}_1\mathbf{c}_2^H|=0$ or $\frac{1}{T_{(i,j)}}$.
		\end{itemize}
		Hence, we have that {\small \[I_{\mathrm{max}}(\mathcal{C}^{(3)})=
			\frac{p^{\frac{n+m}{2}}}{(p^m-2)(p^n-T)}.
			\]} Obviously, $N=p^n(p^{2m}-p^m+1)$, then
		{\small \[\frac{I_{\mathrm{max}}(\mathcal{C}^{(3)})}{I_{\mathrm{W}}}=\sqrt{\frac{p^{n+2m}(p^n(p^{2m}-p^m+1)-1)}{(p^m-1)^4(p^{n-m}-\varepsilon p^{\frac{n}{2}-m})^2(p^m-2)^2}}.\]}
		Since $m\le \frac{n}{2}$, then $\frac{I_{\mathrm{max}}(\mathcal{C}^{(3)})}{I_{\mathrm{W}}}\rightarrow 1$ if $p\rightarrow +\infty$ or $m\rightarrow +\infty$, which implies that the codebook $\mathcal{C}^{(3)}$ asymptotically achieves the Welch bound. By Lemma \ref{Le6} and the above discussion, we can easily obtain the distribution of cross-correlation amplitudes of $\mathcal{C}^{(3)}$.
	\end{proof}
	\section{Conclusion}
	In this paper, by utilizing additive and multiplicative characters, we presented five constructions of codebooks from vectorial dual-bent functions, and the resulting codebooks asymptotically achieve the Welch bound. For these codebooks, the maximum cross-correlation amplitudes and the distributions of the cross-correlation amplitudes were explicitly determined. Compared with the known asymptotically optimal codebooks with respect to the Welch bound constructed in \cite{Ding4,Heng3,Heng2,Heng1,Hong,Hu,Lu,Luo1,Tian,Wu,Yan,Yin,Yu,Zhang,Zhou2,Zhou1}, our codebooks have different parameters. To conclude this paper, we list the following two open problems.
	\begin{itemize}
		
		\item[$\bullet$]In Theorem \ref{Th1}, we restricted the range of 
		$|I|$. Is this range tight for the construction of asymptotically optimal codebooks? In particular, does there exist a value of 
		$|I|$ outside this range that also yields an asymptotically optimal codebook?
		\item [$\bullet$] In Section $\mathrm{\RomanNum{3}}$, based on Equation (1), we selected different sets and obtained some asymptotically optimal codebooks. It is a natural question whether other sets can be used in (1) to construct asymptotically optimal codebooks.
	\end{itemize}

\end{document}